\documentclass[runningheads]{llncs} 
\usepackage{pipa} 
\usepackage{version} 
\usepackage[hidelinks]{hyperref}
\hypersetup{colorlinks=true,linkcolor=blue,citecolor=blue,urlcolor=blue}
\usepackage{orcidlink}

\newcommand{\doi}[1]{{doi:\href{http://doi.org/#1}{\nolinkurl{#1}}}}

\spnewtheorem{sdef}{Definition}{\bfseries}{\rmfamily}
\spnewtheorem{sass}{Assumption}{\bfseries}{\rmfamily}
\spnewtheorem{snot}{Notation}{\bfseries}{\rmfamily}
\spnewtheorem{scnv}{Convention}{\bfseries}{\rmfamily}
\spnewtheorem{sexa}{Example}{\bfseries}{\rmfamily}
\spnewtheorem{sclm}{Claim}{\bfseries}{\rmfamily}

\title{Probabilistic Imperative Process Algebra}
\author{C.A. Middelburg\,\orcidlink{0000-0002-8725-0197}}
\institute{Informatics Institute, Faculty of Science, University of
           Amsterdam \\
           Science Park~900, 1098~XH Amsterdam, the Netherlands \\
           \href{mailto:C.A.Middelburg@uva.nl}{C.A.Middelburg@uva.nl}}

\titlerunning{Probabilistic Imperative Process Algebra}
\authorrunning{C.A. Middelburg}

\begin{document}
\maketitle

\begin{abstract}
In a previous paper, a process algebra based on ACP (Algebra of 
Communicating Processes) was proposed in which processes involving data 
can be handled by means of features originating from imper\-ative 
programming.
In this paper, an extension of that process alge\-bra with probabilistic
choice operators is presented that rests on the principle that 
probabilistic choices are always resolved before choices involved in 
alternative composition and parallel composition are resolved.
This extension is devised, among other things, to be used for modeling 
and analyzing algorithms that are important in the area of distributed 
computing.
Many canonical problems in that area call for a probabilistic algorithm.
In this paper, a probabilistic algorithm for the leader election problem
is modeled using the presented process algebra.
\keywords
{imperative process algebra \and probabilistic choice \and
 abstraction \and branching bisimulation \and leader election}
\begin{classcode}
D.1.3, D.2.4, F.1.2, F.3.1, C.2.1.
\end{classcode}
\end{abstract}

\section{Introduction}
\label{sect-intro}

A contemporary computer-based system usually carries out a process that 
is in ongoing interaction with its environment and in which data plays 
a crucial role.
That is, data change in the course of the process, the process proceeds 
at certain stages in a way that depends on changing data, and the 
interaction of the process with other processes consists of 
communication of data.
In~\cite{Mid21a}, an extension of ACP~\cite{BK84b} was introduced that
was devised to be used for modeling and analyzing such systems.
The extension concerned is called \deACPet. 
In~\cite{NP97a}, the term imperative process algebra was coined for 
process algebras like \deACPet.

However, contemporary computer-based systems are often distributed 
systems.
These systems are usually faced with problems typical of distributed 
systems (see e.g.~\cite{AW04a}). 
For many canonical problems in this area, such as the leader election 
problem and the consensus problem, it is desirable or even required to 
solve them using a probabilistic algorithm.
This calls for an extension of \deACPet\ in which probabilistic 
processes can be handled.
In this paper, an extension\linebreak[2] of \deACPet\ with probabilistic 
choice operators, called \depACPet\ (probabilistic \deACPet), is 
presented.

The main motivation for the development of the presented process algebra 
was the wish to be able to use an imperative process algebra for 
modeling and analyzing the probabilistic leader election algorithms for 
anonymous rings given in~\cite{FP06a} and~\cite{BFPP08a}.
In those papers, probabilistic leader election algorithms are modeled 
and analyzed using the $\mu$CRL specification language and 
toolset~\cite{BFGLLP01a}.
The $\mu$CRL specification language is based on ACP, but cannot handle
probabilistic processes.
Moreover, data must be dealt with in a non-imperative way, namely by 
means of parameterized processes.
Therefore, I wondered whether a probabilistic imperative process algebra 
could be devised that is better suited for modeling and analyzing those 
algorithms.
The suitability of the presented process algebra for modeling the 
above-mentioned leader election algorithms is demonstrated in this 
paper.
The suitability of the presented process algebra for analyzing them 
remains the subject of future work.

The extension of \deACPet\ with probabilistic choice operators rests on 
the principle that probabilistic choices are always resolved before 
choices involved in alternative composition and parallel composition are 
resolved.
This principle is also followed in~\cite{AG09a,Geo11a}.
However, in \depACPet, we take functions whose range is the carrier of 
a signed cancellation meadow instead of a field as probability measures, 
add probabilistic choice operators for the probabilities $0$ and $1$, 
and have an additional axiom because of the inclusion of these 
operators.
The probabilistic choice operators for the probabilities $0$ and $1$ 
cause no problem because a meadow has a total multiplicative inverse 
operation where the multiplicative inverse of zero is zero.
Because of this property, we could reduce the number of rules for the 
operational semantics of \depACPet\ and replace all negative premises by 
positive premises in the remaining rules. 

This paper is organized as follows.
First, the theory of signed cancellation meadows is briefly summarized  
(Section~\ref{sect-meadows}).
Next, the algebraic theory \mbox{\depACPet}\ is introduced
(Sections~\ref{sect-pACPet} and~\ref{sect-depACPet}).
After that, the extension of \depACPet\ with guarded linear recursion is 
treated (Section~\ref{sect-depACPetr}).
Then, a structural operational semantics of the resulting 
theory is presented and a notion of rooted branching bisimulation 
equivalence based on it is defined (Section~\ref{sect-semantics}).
After that, a soundness result with respect to rooted branching 
bisimulation equivalence for the axiom system of \depACPet\ is 
presented (Section~\ref{sect-sound-compl}).
Thereafter, the probabilistic algorithm for the leader election problem
given in~\cite{BFPP08a} is modeled using \depACPet\
(Section~\ref{sect-example}).
Finally, some concluding remarks are made 
(Section~\ref{sect-conclusions}).

\section{Signed Cancellation Meadows}
\label{sect-meadows}

Later in this paper, we will take functions whose range is the carrier 
of a signed cancellation meadow as probability measures.
Therefore, we briefly summarize the theory of signed cancellation 
meadows in this section.

In~\cite{BT07a}, meadows are proposed as alternatives for fields with a
purely equational axiomatization.
Meadows are commutative rings with a multiplicative identity element and 
a total multiplicative inverse operation where the multiplicative 
inverse of zero is zero.
Fields whose multiplicative inverse operation is made total by imposing 
that the multiplicative inverse of zero is zero are called 
zero-totalized fields.
A zero-totalized field is a meadow that satisfies the 
\emph{cancellation axiom} 
\begin{ldispl}
x \neq 0 \Land x \mul y = x \mul z \Limpl y = z
\end{ldispl}%
and the \emph{separation axiom} 
\begin{ldispl}
0 \neq 1\;.
\end{ldispl}%

Meadows that satisfy the cancellation axiom are called cancellation 
meadows.
Signed cancellation meadows are cancellation meadows expanded with a
signum operation.
The signum operation makes it possible that the predicates $<$ and 
$\leq$ are defined (see below).

\begin{sdef}
The signature of signed cancellation meadows consists of the following
constants and operators:
\begin{itemize}
\item
the \emph{additive identity} constant $0$;
\item
the \emph{multiplicative identity} constant $1$;
\item
the binary \emph{addition} operator ${} +$ {};
\item
the binary \emph{multiplication} operator ${} \mul {}$;
\item
the unary \emph{additive inverse} operator $- {}$\,;
\item
the unary \emph{multiplicative inverse} operator ${}\minv$\,;
\item
the unary \emph{signum} operator $\sign$.
\end{itemize}
\end{sdef}

Terms are build as usual.
We use prefix, infix, and postfix notation as usual.
We also use the usual precedence convention.
Subtraction and division are introduced as abbreviations:
$t - t'$ abbreviates $t + (-t')$ and
$t / t'$ abbreviates $t \mul ({t'}\minv)$.

\begin{sdef}
The axioms of a signed cancellation meadow are the equations in 
Tables~\ref{eqns-meadow} and~\ref{eqns-signum} and the above-mentioned
cancellation axiom.
\begin{table}[!t]
\caption
{Axioms of a meadow}
\label{eqns-meadow}
\begin{eqntbl}
\begin{eqncol}
(x + y) + z = x + (y + z)                                             \\
x + y = y + x                                                         \\
x + 0 = x                                                             \\
x + (-x) = 0
\end{eqncol}
\qquad\quad
\begin{eqncol}
(x \mul y) \mul z = x \mul (y \mul z)                                 \\
x \mul y = y \mul x                                                   \\
x \mul 1 = x                                                          \\
x \mul (y + z) = x \mul y + x \mul z
\end{eqncol}
\qquad\quad
\begin{eqncol}
(x\minv)\minv = x                                                     \\
x \mul (x \mul x\minv) = x                           
\end{eqncol}
\end{eqntbl}
\end{table}

\begin{table}[!t]
\caption{Additional axioms for the signum operator}
\label{eqns-signum}
\begin{eqntbl}
\begin{eqncol}
\sign(x / x) = x / x                                                  \\
\sign(1 - x / x) = 1 - x / x                                          \\
\sign(-1) = -1
\end{eqncol}
\qquad\quad
\begin{eqncol}
\sign(x\minv) = \sign(x)                                              \\
\sign(x \mul y) = \sign(x) \mul \sign(y)                              \\
(1 - \frac{\sign(x) - \sign(y)}{\sign(x) - \sign(y)}) \mul
(\sign(x + y) - \sign(x)) = 0
\end{eqncol}
\end{eqntbl}
\end{table}
\end{sdef}

\begin{sdef}
The predicates $<$ and $\leq$ are defined in signed cancellation meadows
as follows: 
\begin{ldispl}
x < y \Liff \sign(y - x) = 1\;,
\\
x \leq y \Liff \sign(\sign(y - x) + 1) = 1\;.
\end{ldispl}%
\end{sdef}

It is easy to see that
\begin{ldispl}
0 \leq x \leq 1 \Liff
 \sign(\sign(x) + 1) \mul \sign(\sign(1 - x) + 1) = 1\;.
\end{ldispl}%
We will use this equivalence below to describe the set of probabilities.

\sloppy
In~\cite{BP13a}, Kolmogorov's probability axioms for finitely additive 
probability spaces are rephrased for the case where probability measures 
are functions whose range is the carrier of a signed cancellation 
meadow.

\section{Probabilistic \ACP\ with Empty Process and Silent Step}
\label{sect-pACPet}

In this section, the process algebra \pACPet\ is presented.
\pACPet\ is the probabilistic process algebra \pACP\ presented 
in~\cite{Mid19a} extended with the termination constant $\ep$, the 
silent step constant $\tau$, and the abstraction operators $\abstr{I}$ 
as in the extension \ACPet\ of \ACP~\cite{BK84b} presented 
in~\cite{BW90}.
In Section~\ref{sect-depACPet}, \pACPet\ is extended with features that 
are relevant to processes in which data are involved.

\begin{sass}
\label{assumption-basic-actions}
It is assumed that a fixed but arbitrary finite set $\Act$ 
of \emph{basic actions}, with $\tau,\dead,\ep \not\in \Act$, and a fixed 
but arbitrary commutative and associative \emph{communication} function 
$\funct{\commf}
 {(\Act \union \set{\tau,\dead}) \x (\Act \union \set{\tau,\dead})}
 {(\Act \union \set{\tau,\dead})}$, 
such that $\commf(\tau,a) = \dead$ and $\commf(\dead,a) = \dead$
for all $a \in \Act \union \set{\tau,\dead}$, have been given.
\end{sass}
Basic actions are taken as atomic processes.
For any two basic actions $a$ and $b$, $\commf(a,b)$ is considered to be 
the basic action that results from performing them synchronously if they 
can be performed synchronously, and $\dead$ otherwise.
\begin{snot}
We write $\Actt$ for $\Act \union \set{\tau}$ and 
$\Acttd$ for $\Act \union \set{\tau,\dead}$.
\end{snot}

\begin{sass}
\label{assumption-meadow}
It is assumed that a fixed but arbitrary signed cancellation meadow 
$\fM$ has been given.
\end{sass}
\begin{snot}
We denote the interpretations of the constants and operators from the 
signature of signed cancellation meadows in $\fM$ by the constants and 
operators themselves.
\end{snot}
\begin{snot}
We write $\Prob$ for the set
$\set{\pi \in \fM \where
 \sign(\sign(\pi) + 1) \mul \sign(\sign(1 - \pi) + 1) = 1}$
of \emph{probabilities}.
\end{snot}

The algebraic theory \pACPet\ consists of a signature and an axiom 
system.
What is assumed to be given in 
Assumptions~\ref{assumption-basic-actions} and~\ref{assumption-meadow} 
can be considered formal parameters of the algebraic theory \pACPet\ --- 
of which different concretizations yield different instantiations of the 
theory.

\begin{sdef}
The signature of the algebraic theory \pACPet\ consists of the following 
sorts, constants, and operators: 
\begin{itemize}
\item
the sort $\Proc$ of \emph{processes};
\item
for each $a \in \Act$, the \emph{basic action} constant 
$\const{a}{\Proc}$;
\item
the \emph{silent step} constant $\const{\tau}{\Proc}$;
\item
the \emph{inaction} constant $\const{\dead}{\Proc}$;
\item
the \emph{termination} constant $\const{\ep}{\Proc}$;
\item
the binary \emph{alternative composition} operator 
$\funct{\altc}{\Proc \x \Proc}{\Proc}$;
\item
the binary \emph{sequential composition} operator 
$\funct{\seqc}{\Proc \x \Proc}{\Proc}$;
\item
the binary \emph{parallel composition} operator 
$\funct{\parc}{\Proc \x \Proc}{\Proc}$;
\item
the binary \emph{left merge} operator 
$\funct{\leftm}{\Proc \x \Proc}{\Proc}$;
\item
the binary \emph{communication merge} operator 
$\funct{\commm}{\Proc \x \Proc}{\Proc}$;
\item
the unary \emph{termination} operator 
$\funct{\encapa}{\Proc}{\Proc}$;
\item
for each $H \subseteq \Act$, 
the unary \emph{encapsulation} operator 
$\funct{\encap{H}}{\Proc}{\Proc}$;
\item
for each $I \subseteq \Act$, 
the unary \emph{abstraction} operator 
$\funct{\abstr{I}}{\Proc}{\Proc}$;
\item
for each $\pi \in \Prob$,
the binary \emph{probabilistic choice} operator 
$\funct{\paltc{\pi}}{\Proc \x \Proc}{\Proc}$.
\end{itemize}
\end{sdef}
\begin{sass}
It is assumed that there is a countably infinite set $\cX$ of variables 
of sort $\Proc$, which contains $x$, $y$ and $z$.
\end{sass}
Terms are built as usual.
\begin{snot}
Infix notation is used for the binary operators.

The following precedence conventions are used to reduce the need for
parentheses: the operator $\altc$ has a lower precedence than all other 
binary operators and the operator $\seqc$ has a higher precedence than 
all other binary operators.

The following associativity convention is used to reduce the need for
parentheses: the binary operators $\altc$ and $\seqc$ are right 
associative.
\end{snot}
Because the axiom system of \pACPet\ includes associative laws for
the operators $\altc$ and $\seqc$, the choice between left associativity 
and right associativity seems arbitrary.
However, right associativity is more convenient when deriving equations 
from the axioms of \pACPet.

\begin{scnv}
In explanations, we sloppily say ``$t$'', where $t$ is a closed term of 
sort $\Proc$, instead of ``the process denoted by $t$''.
\end{scnv}

Let  $t$ and $t'$ be closed \pACPet\ terms, $a \in \Act$,
$H,I \subseteq \Act$, and $\pi \in \Prob$.
Then the constants and operators of \pACPet\ can be explained as 
follows:
\begin{itemize}
\item
$a$ performs the observable action $a$ and after that terminates 
successfully;
\item
$\tau$ performs the unobservable action $\tau$ and after that terminates 
successfully;
\item
$\dead$ cannot do anything, it cannot even terminate successfully;
\item
$\ep$ terminates successfully without performing any action.
\item
$t \altc t'$ behaves as either $t$ or $t'$;
 \item
$t \seqc t'$ behaves as $t$ and $t'$ in sequence;
\item
$t \parc t'$  behaves as $t$ and $t'$ in parallel; 
\item
$t \leftm t'$ behaves the same as $t \parc t'$, except that it starts 
with performing an action of $t$;
\item
$t \commm t'$ behaves the same as $t \parc t'$, except that it starts 
with performing an action of $t$ and an action of $t'$ synchronously;
\item
$\encapa(t)$ terminates successfully without performing any action if 
$t$ has the option to terminate successfully and cannot do anything 
otherwise;
\item
$\encap{H}(t)$ behaves the same as $t$, except that actions from $H$ are 
blocked from being performed;
\item
$\abstr{I}(t)$ behaves the same as $t$, except that actions from $I$ are 
turned into the unobservable action $\tau$;
\item
$t \paltc{\pi} t'$ behaves as $t$ with probability $\pi$ and as $t'$ 
with probability $1 - \pi$.
\end{itemize}
Here ``behaves as $t$ and $t'$ in parallel'' means that 
(a)~each time an action is performed, either a next action of $t$ is 
performed or a next action of $t'$ is performed or a next action of $t$ 
and a next action of $t'$ are performed synchronously and
(b)~successful termination may take place at any time that both $t$ and 
$t'$ can terminate successfully.

In the case of $t \altc t'$, the choice between $t$ and $t'$ is resolved 
at the instant that one of them performs its first action or terminates 
successfully without performing any action, and not before.
In the case of $t \paltc{\pi} t'$, the choice between $t$ and $t'$ is 
resolved before one of them performs its first action or terminates 
successfully without performing any action.

The operators $\leftm$, $\commm$, and $\encapa$ are of an auxiliary 
nature.
They make a finite axiomatization of \pACPet\ possible.

\begin{sdef}
The axiom system of the algebraic theory \pACPet\ consists of the 
equations and conditional equations presented in 
Tables~\ref{axioms-pACPet-1} and~\ref{axioms-pACPet-2}.
\begin{table}[!t]
\caption{Axioms of \pACPet\ (part~1)}
\label{axioms-pACPet-1}
\begin{eqntbl}
\begin{axcol}
x \altc y = y \altc x                                & & \axiom{A1}   \\
(x \altc y) \altc z = x \altc (y \altc z)            & & \axiom{A2}   \\
a \altc a = a                                        & & \axiom{A3'}  \\
\ep \altc \ep = \ep                                  & & \axiom{A3''} \\
(x \altc y) \seqc z = x \seqc z \altc y \seqc z      & & \axiom{A4}   \\
(x \seqc y) \seqc z = x \seqc (y \seqc z)            & & \axiom{A5}   \\
x \altc \dead = x                                    & & \axiom{A6}   \\
\dead \seqc x = \dead                                & & \axiom{A7}   \\
x \seqc \ep = x                                      & & \axiom{A8}   \\
\ep \seqc x = x                                      & & \axiom{A9} 
\eqnsep
x = x \altc x \Land y = y \altc y \Limpl \\ \quad\;\;
x \parc y = x \leftm y \altc y \leftm x \altc x \commm y \altc
\encapa(x) \seqc \encapa(y)                          & & \axiom{CM1E'} \\
\ep \leftm x = \dead                                 & & \axiom{CM2E} \\
\alpha \seqc x \leftm y = \alpha \seqc (x \parc y)   & & \axiom{CM3}  \\
(x \altc y) \leftm z = x \leftm z \altc y \leftm z   & & \axiom{CM4}  \\
\ep \commm x = \dead                                 & & \axiom{CM5E} \\
x \commm \ep = \dead                                 & & \axiom{CM6E} \\
a \seqc x \commm b \seqc y = \commf(a,b) \seqc (x \parc y) 
                                                     & & \axiom{CM7}  \\
(x \altc y) \commm z = x \commm z \altc y \commm z   & & \axiom{CM8}  \\
x \commm (y \altc z) = x \commm y \altc x \commm z   & & \axiom{CM9}  \\
\dead \commm x = \dead                               & & \axiom{CM10} \\
x \commm \dead = \dead                               & & \axiom{CM11} \\
a \commm b = \commf(a,b)                             & & \axiom{CM12} 
\eqnsep
\encapa(\ep) = \ep                                   & & \axiom{TE1}  \\
\encapa(\alpha) = \dead                              & & \axiom{TE2}  \\
\encapa(x \altc y) = \encapa(x) \altc \encapa(y)     & & \axiom{TE3}  \\
\encapa(x \seqc y) = \encapa(x) \seqc \encapa(y)     & & \axiom{TE4}
\eqnsep
\encap{H}(\ep) = \ep                                   & & \axiom{D0} \\
\encap{H}(\alpha) = \alpha      & \mif \alpha \notin H   & \axiom{D1} \\ 
\encap{H}(\alpha) = \dead       & \mif \alpha \in H      & \axiom{D2} \\
\encap{H}(x \altc y) = \encap{H}(x) \altc \encap{H}(y) & & \axiom{D3} \\
\encap{H}(x \seqc y) = \encap{H}(x) \seqc \encap{H}(y) & & \axiom{D4}
\eqnsep
\abstr{I}(\ep) = \ep                                   & & \axiom{T0} \\
\abstr{I}(\alpha) = \alpha      & \mif \alpha \notin I   & \axiom{T1} \\
\abstr{I}(\alpha) = \tau        & \mif \alpha \in I      & \axiom{T2} \\
\abstr{I}(x \altc y) = \abstr{I}(x) \altc \abstr{I}(y) & & \axiom{T3} \\
\abstr{I}(x \seqc y) = \abstr{I}(x) \seqc \abstr{I}(y) & & \axiom{T4} 
\end{axcol}
\end{eqntbl}
\end{table}
\begin{table}[!t]
\caption{Axioms of \pACPet\ (part~2)}
\label{axioms-pACPet-2}
\begin{eqntbl}
\begin{axcol}
x \paltc{\pi} y = y \paltc{1{-}\pi} x                  & & \axiom{pA1}  \\
(x \paltc{\pi} y) \paltc{\rho} z = 
x \paltc{\pi{\mul}\rho}
\smash{(y \paltc{\frac{(1{-}\pi){\mul}\rho}{1{-}\pi{\mul}\rho}} z)}
                                                       & & \axiom{pA2}  \\
x \paltc{\pi} x = x                                    & & \axiom{pA3}  \\
(x \paltc{\pi} y) \seqc z = x \seqc z \paltc{\pi} y \seqc z
                                                       & & \axiom{pA4}  \\
(x \paltc{\pi} y) \altc z = (x \altc z) \paltc{\pi} (y \altc z)
                                                       & & \axiom{pA5}  \\
x \paltc{1} y = x                                      & & \axiom{pA6}  
\eqnsep
(x \paltc{\pi} y) \parc z = (x \parc z) \paltc{\pi} (y \parc z)
                                                       & & \axiom{pCM1} \\                                                               
x \parc (y \paltc{\pi} z) = (x \parc y) \paltc{\pi} (x \parc z)
                                                       & & \axiom{pCM2} \\
(x \paltc{\pi} y) \leftm z = (x \leftm z) \paltc{\pi} (y \leftm z)
                                                       & & \axiom{pCM3} \\                                                               
x \leftm (y \paltc{\pi} z) = (x \leftm y) \paltc{\pi} (x \leftm z)
                                                       & & \axiom{pCM4} \\
(x \paltc{\pi} y) \commm z = (x \commm z) \paltc{\pi} (y \commm z)
                                                       & & \axiom{pCM5} \\                                                               
x \commm (y \paltc{\pi} z) = (x \commm y) \paltc{\pi} (x \commm z)
                                                       & & \axiom{pCM6} 
\eqnsep
\encapa(x \paltc{\pi} y)   = \encapa(x) \paltc{\pi} \encapa(y)
                                                       & & \axiom{pTE}   \\
\encap{H}(x \paltc{\pi} y) = \encap{H}(x) \paltc{\pi} \encap{H}(y)
                                                       & & \axiom{pD}   \\
\abstr{I}(x \paltc{\pi} y) = \abstr{I}(x) \paltc{\pi} \abstr{I}(y)
                                                       & & \axiom{pT}  
\eqnsep 

x = x \altc x \Land y = y \altc y \Land \encapa(x \altc y) = \dead \Limpl 
\\ \quad\;\;                                                        
\alpha \seqc ((\tau \seqc (x \altc y) \altc x) \paltc{\pi} z) = 
\alpha \seqc ((x \altc y) \paltc{\pi} z)
                                                       & & \axiom{pBE} 
\end{axcol}
\end{eqntbl}
\end{table}
In these tables, 
$a$, $b$, and $\alpha$ stand for arbitrary constants of \pACPet\ other 
than $\ep$,\, 
$H$ and $I$ stand for arbitrary subsets of $\Act$, and 
$\pi$ and $\rho$ stand for arbitrary probabilities from $\Prob$.
\end{sdef}
A3$'$, CM3, CM7, CM12, TE2, D0--D4, T0--T4, pA1--pA5, pCM1--pCM6, pD, 
pT, and pBE are actually axiom schemas.
\begin{scnv}
In this paper, axiom schemas will usually be referred to as axioms.
\end{scnv}

The occurrence of the strange-looking term $\encapa(x) \seqc \encapa(y)$ 
in axiom CM1E deserves some explanation.
This term is needed to handle successful termination in the presence of 
the constant $\ep$:
it stands for the process that behaves the same as $\ep$ if both $x$ and 
$y$ stand for a process that has the option to behave the same as $\ep$ 
and it stands for the process that behaves the same as $\dead$ 
otherwise.

Axiom pBE is axiom BE of \ACPet~\cite{Mid21a} generalized to the 
probabilistic setting.
The consequent of pBE holds only if:
\begin{itemize}
\item 
$x$ and $y$ stand for processes that do not have to resolve a 
probabilistic choice before they can perform their first action (which 
is expressed by the first two conjuncts of the antecedent);
\item
$x$ and $y$ stand for processes that do not have the option to terminate 
successfully without performing any action (which is expressed by the 
third conjunct of the antecedent).
\end{itemize}
This means that pBE does not allow the removal of a silent step if it 
is immediately followed by a process that has to resolve a probabilistic 
choice before it can perform its first action or it is immediately 
followed by a process that has the option to terminate successfully 
without performing any action.
Without the antecedent of pBE, this axiom would not be sound with 
respect to rooted branching bisimulation equivalence as defined in 
Section~\ref{sect-sound-compl}.

In the sequel, the notation $\Altc{i=1}{n} t_i$, where $n \geq 0$, will 
be used for right-nested alternative compositions.
\begin{snot}
For each $n \in \Nat$, 
the term $\Altc{i = 1}{n} t_i$ is defined by induction on $n$ as 
follows:
\begin{ldispl}
\Altc{i = 1}{0} t_i = \dead
\quad \mathrm{and} \quad
\Altc{i = 1}{1} t_i = t_1  
\quad \mathrm{and} \quad
\Altc{i = 1}{n + 2} t_i = t_1 \altc \Altc{i = 1}{n + 1} t_{i+1}\;.
\end{ldispl}%
\end{snot}

In the sequel, the notation $\Parc{i=1}{n} t_i$, where $n \geq 0$, will 
be used for right-nested parallel compositions.
\begin{snot}
For each $n \in \Nat$, 
the term $\Parc{i = 1}{n} t_i$ is defined by induction on $n$ as 
follows:
\begin{ldispl}
\Parc{i = 1}{0} t_i = \dead
\quad \mathrm{and} \quad
\Parc{i = 1}{1} t_i = t_1  
\quad \mathrm{and} \quad
\Parc{i = 1}{n + 2} t_i = t_1 \parc \Parc{i = 1}{n + 1} t_{i+1}\;.
\end{ldispl}%
\end{snot}

In the sequel, the notation $\PRC{i = 1}{n}{\pi_i}{t_i}$, where 
$n \geq 1$ and $\sum_{i \leq n} \pi_i = 1$, will be used for 
right-nested probabilistic choices.
\begin{snot}
For each $n \in \Natpos$,%
\footnote
{We write $\Natpos$ for the set $\set{n \in \Nat \where n \geq 1}$ of 
 positive natural numbers.} 
the term $\PRC{i = 1}{n}{\pi_i}{t_i}$ is defined by induction on $n$ as 
follows:
\begin{ldispl}
\PRC{i = 1}{1}{\pi_i}{t_i}   = t_1
\quad \mathrm{and} \quad
\PRC{i = 1}{n+1}{\pi_i}{t_i} =
t_1 \paltc{\pi_1} 
(\PRC{i = 1}{n}{\frac{\pi_{i+1}}{1 - \pi_1}}{t_{i+1}})\;.
\end{ldispl}%
\end{snot}
The process denoted by $\PRC{i = 1}{n + 1}{\pi_i}{t_i}$ behaves as 
the process denoted by $t_1$ with probability $\pi_1$, \ldots,  
like the process denoted by $t_{n+1}$ with probability $\pi_{n+1}$.

\begin{sexa}
The process of throwing a die once can be described as follows:
\begin{ldispl}
\throw{1} \paltc{1/6} (\throw{2} \paltc{1/5} (\throw{3} \paltc{1/4}
(\throw{4} \paltc{1/3} (\throw{5} \paltc{1/2} \throw{6}))))\;,
\end{ldispl}%
or using the notation just introduced:
\begin{ldispl}
\PRC{i = 1}{6}{1/6}{\throw{i}}\;.
\end{ldispl}%
\end{sexa}

\section{Imperative \pACPet}
\label{sect-depACPet}

In this section, \depACPet, imperative \pACPet, is presented.
This extension of \pACPet\ has its origin in~\cite{Mid21a}.
It has features that are relevant to processes in which data are 
involved, such as data parameterized actions (to deal with process 
interactions with data transfer), assignment actions (to deal with 
data that change in the course of a process), and guarded commands 
(to deal with processes that only take place if some data-dependent 
condition holds). 

\begin{sass}
\label{assumption-data}
It is assumed that the following has been given with 
respect to data:
\begin{itemize}
\item
a many-sorted signature $\sig_\gD$ that includes:
\begin{itemize}
\item
sorts $\Data_1,\ldots,\Data_n$ ($n > 0$) of \emph{data};
\item 
a sort $\Bool$ of \emph{booleans};
\item
constants and operators such that there exist closed terms of the sorts 
$\Data_1,\ldots,\Data_n$;
\item
constants $\Btrue$ and $\Bfalse$ of sort $\Bool$;
\end{itemize}
\item
a minimal algebra $\gD$ of the signature $\sig_\gD$ in which: 
\begin{itemize}
\item
the cardinality  of the carrier of sort $\Bool$ is $2$;
\item
the equation $\Btrue = \Bfalse$ does not hold.
\end{itemize}
\end{itemize}
\end{sass}
The sort $\Bool$ is assumed to be given primarily in order to make it 
possible for operators to serve as predicates.
\begin{snot}
We write $\DataVal_1$, \ldots, $\DataVal_n$, for the sets of 
all closed terms over the signature $\sig_\gD$ that are of sort 
$\Data_1$, \ldots, $\Data_n$, respectively.
\end{snot}

In \deACPet~\cite{Mid21a} --- of which \depACPet\ is an extension with
probabilistic choice operators --- only one sort of data is assumed to 
be given.
For convenience in applications, the restriction to one sort of data has
been removed in \depACPet.

\begin{sass}
\label{assumption-flex-vars}
It is assumed that finite or countably infinite sets $\ProgVar_1$, 
\ldots, $\ProgVar_n$ of \emph{flexible variables of sort $\Data_1$, 
\ldots, $\Data_n$}, respectively, have been given and that those sets 
are mutually disjoint.
\end{sass}
A flexible variable is a variable whose value may change in the course 
of a process.%
\footnote
{The term flexible variable is used for this kind of variables in 
 e.g.~\cite{Sch97a,Lam94a}.} 
Typical examples of flexible variables are the program variables known 
from imperative programming.

\begin{sdef}
An \emph{evaluation map} is a function $\sigma$ 
from $\Union_{i \in \set{1,\ldots,n}} \ProgVar_i$ 
to $\Union_{i \in \set{1,\ldots,n}} \DataVal_i$ 
such that, for all $i,j \in \set{1,\ldots,n}$ with $i \neq j$,
for all $v \in \ProgVar_i$ and $e \in \DataVal_j$, $\sigma(v) \neq e$.
\end{sdef}
\begin{snot}
We write $\EvalMap$ for the set of all evaluation maps.
\end{snot}

Evaluation maps are intended to provide the data values assigned to 
flexible variables when an \depACPet\ term of a sort of data is 
evaluated.
However, in order to fit better in an algebraic setting, they provide 
closed terms over the signature $\sig_\gD$ that denote those data 
values instead.
The requirement that $\gD$ is a minimal algebra guarantees that each 
data value can be represented by a closed term. 

Below, the signature of \depACPet\ is introduced.
The signature of \depACPet\ includes a variable-binding operator.
The formation rules for \depACPet\ terms are the usual ones for the 
many-sorted case (see e.g.~\cite{ST99a,Wir90a}) and in addition the 
following rule:
\begin{itemize}
\item
if $O$ is a variable-binding operator 
$\funct{O}{S_1 \x \ldots \x S_n}{S}$ that binds a variable of sort~$S'$,
$t_1,\ldots,t_n$~are terms of sorts $S_1,\ldots,S_n$, respectively, and 
$X$ is a variable of sort $S'$, then $O X (t_1,\ldots,t_n)$ is a term of 
sort $S$.
\end{itemize}
An extensive formal treatment of the phenomenon of variable-binding 
operators can be found in~\cite{PS95a}.

The signature and axiom system of the algebraic theory \depACPet\ 
are extensions of the signature and axiom system of \pACPet.
What is assumed to be given in 
Assumptions~\ref{assumption-basic-actions}, \ref{assumption-meadow},
\ref{assumption-data}, and~\ref{assumption-flex-vars} can be considered 
formal parameters of the algebraic theory \depACPet\ --- of which 
different concretizations yield different instantiations of the theory.

\begin{sdef}
The signature of the algebraic theory \depACPet\ consists of 
the sorts, constants, and operators from the signature $\sig_\gD$,
the sorts, constants, and operators from the signature of \deACPet, and
the following sorts, constants, and operators:
\begin{itemize}
\item
the sort $\Cond$ of \emph{conditions};
\item
for each $i \in \set{1,\ldots,n}$, for each $v \in \ProgVar_i$, 
the \emph{flexible variable} constant $\const{v}{\Data_i}$;
\item
for each $i \in \set{1,\ldots,n}$,
the binary \emph{equality} operator
$\funct{\Leq}{\Data_i \x \Data_i}{\Cond}$;%
\footnote
{The overloading of $=$ can be trivially resolved if $\sig_\gD$ is
 without overloaded symbols.}
\item
the binary \emph{equality} operator
$\funct{\Leq}{\Bool \x \Bool}{\Cond}$;
\item
the \emph{falsity} constant $\const{\False}{\Cond}$;
\item
the unary \emph{negation} operator $\funct{\Lnot}{\Cond}{\Cond}$;
\item
the binary \emph{disjunction} operator 
$\funct{\Lor}{\Cond \x \Cond}{\Cond}$;
\item
the unary variable-binding \emph{existential quantification} operator 
$\funct{\exists}{\Cond}{\Cond}$ that binds a variable of sort $\Data$; 
\item
for each $i \in \set{1,\ldots,n}$, for each $a \in \Act$,
the unary \emph{data parameterized action} operator
$\funct{a}{\Data_i}{\Proc}$;
\item
for each $i \in \set{1,\ldots,n}$, for each $v \in \ProgVar_i$, 
the unary \emph{assignment action} operator
$\funct{\assop{v}\,}{\Data_i}{\Proc}$;
\item
the binary \emph{guarded command} operator 
$\funct{\gc\,}{\Cond \x \Proc}{\Proc}$;
\item
for each $\sigma \in \EvalMap$, 
the unary \emph{evaluation} operator 
$\funct{\eval{\sigma}}{\Proc}{\Proc}$.
\end{itemize}
\end{sdef}
\begin{sass}
It is assumed that there are countably infinite sets of variables of 
sort $\Data_1$, \ldots, $\Data_n$, and $\Cond$ and that the sets of 
variables of sort $\Data_1$, \ldots, $\Data_n$, $\Cond$, and $\Proc$ 
are mutually disjoint and disjoint from 
$\Union_{i \in \set{1,\ldots,n}} \ProgVar_i$.
\end{sass}
\begin{snot}
Infix notation is also used for the additional binary operators.
\end{snot}
\begin{snot}
The notation $\ass{v}{e}$, where $v \in \ProgVar_i$ and $e$ is a closed
\depACPet\ term of sort $\Data_i$ for some $1 \leq i \leq n$, is used 
for the term $\assop{v}(e)$.
\end{snot}

Because operators with result sort $\Bool$ serve as predicates, the 
following notation is convenient.
\begin{snot}
We write $e_1 \diamond e_2$, 
where $\funct{\diamond}{\Data_i \x \Data_j}{\Bool}$, 
$e_1$ is a closed \depACPet\ term of sort $\Data_i$, and
$e_2$ is a closed \depACPet\ term of sort $\Data_j$,  
for some $i,j \in \set{1,\ldots,n}$, instead of 
$e_1 \diamond e_2 = \Btrue$ where a term of sort $\Cond$ is expected.
\end{snot}

We also use the common logical abbreviations.
\begin{snot}
Let $\phi$ and $\psi$ be \depACPet\ terms of sort $\Cond$ and
let $X$ be a variable of one of the sorts $\Data_1$, \ldots, $\Data_n$.
Then $\True$ stands for $\Lnot \False$,
$\phi \Land \psi$ stands for $\Lnot (\Lnot \phi \Lor \Lnot \psi)$,
$\phi \Limpl \psi$ stands for $\Lnot \phi \Lor \psi$,
$\phi \Liff \psi$ stands for 
$(\phi \Limpl \psi) \Land (\psi \Limpl \phi)$, and
$\Forall{X}{\phi}$ stands for $\Lnot \Exists{X}{\Lnot \phi}$.
\end{snot}

\begin{snot}
We write:
\begin{itemize}
\item
$\ProcTerm$\, for the set of all closed \depACPet\ terms of sort 
$\Proc$;
\item
$\CondTerm$\,\, for the set of all closed \depACPet\ terms of sort 
$\Cond$;
\item
$\DataTerm_1$, \ldots, $\DataTerm_n$ for the sets of all closed 
\depACPet\ terms of sort $\Data_1$, \ldots, $\Data_n$, respectively.
\end{itemize}
\end{snot}

Each term from $\CondTerm$ can be taken as a formula of a first-order 
language with equality of $\gD$ by taking the flexible variable 
constants of sorts $\Data_1$, \ldots, $\Data_n$ as additional variables 
of sorts $\Data_1$, \ldots, $\Data_n$, respectively.
The flexible variable constants are implicitly taken as additional 
variables wherever the context asks for a formula.
In this way, each term from $\CondTerm$ can be interpreted in $\gD$ as a
formula.
The axioms of \depACPet\ (given below) include an equation $\phi = \psi$ 
for each two terms $\phi$ and $\psi$ from $\CondTerm$ for which the 
formula $\phi \Liff \psi$ holds in $\gD$.

Let $t$ be a term from $\ProcTerm$, $\phi$ be a term from $\CondTerm$, 
$e$ be a term from one of the sets $\DataTerm_1$, \ldots, $\DataTerm_n$, 
and $a$ be a basic action from $\Act$. 
Then the additional operators to build terms of sort $\Proc$ can be 
explained as follows:
\begin{itemize}
\item
$a(e)$ performs the data parameterized action $a(e)$ and after that 
terminates successfully;
\item
$\ass{v}{e}$ performs the assignment action $\ass{v}{e}$, whose intended 
effect is the assignment of the result of evaluating $e$ to flexible 
variable $v$, and after that terminates successfully; 
\item
$\phi \gc t$ behaves as $t$ if condition $\phi$ holds and as $\dead$ 
otherwise;
\item
$\eval{\sigma}(t)$ behaves as $t$ after each subterm $e$ of $t$ of one 
of the sorts $\Data_1$, \ldots, $\Data_n$ has been evaluated using the 
evaluation map $\sigma$ updated according to the assignment actions that 
have taken place at the point where the subterm is encountered.
\end{itemize}
Evaluation operators are a variant of state operators 
(see e.g.~\cite{BB88}).

An evaluation map $\sigma$ can be extended homomorphically from flexible
variables to terms from $\DataTerm_1$, \ldots, $\DataTerm_n$, and 
$\CondTerm$.
\begin{snot}
The homomorphic extensions of an evaluation map $\sigma$ from flexible
variables to terms from $\DataTerm_1$, \ldots, $\DataTerm_n$, and 
$\CondTerm$ are denoted by $\sigma$ as well.
\end{snot}
\begin{snot}
We write $\sigma\mapupd{e}{v}$, where $e \in \DataTerm_i$ and 
$v \in \ProgVar_i$ for some $i \in \set{1,\ldots,n}$ and 
$\sigma \in \EvalMap$, for the evaluation map $\sigma'$ defined by 
$\sigma'(v') = \sigma(v')$ if $v' \neq v$ and $\sigma'(v) = e$.
\end{snot}

In the sequel, reference is made to the following two subsets of 
$\ProcTerm$.
\begin{sdef}
\begin{ldispl}
\begin{aeqns}
\AProcDPA  & = & \,  
\Union_{i \in \set{1,\ldots,n}} 
 \set{a(e) \where a \in \Act \Land e \in \DataTerm_i}\;, \\
\AProcTerm & = & 
\set{a \where a \in \Act} \union \AProcDPA \union
\Union_{i \in \set{1,\ldots,n}} 
 \set{\ass{v}{e} \where v \in \ProgVar_i \Land e \in \DataTerm_i}\;.
\end{aeqns}
\end{ldispl}%
\end{sdef}
The elements of $\AProcTerm$ are the terms from $\ProcTerm$ that denote 
the processes that are considered to be atomic.
\begin{snot}
We write $\AProcTermt$ for $\AProcTerm \union \set{\tau}$ and 
$\AProcTermtd$ for $\AProcTerm \union \set{\tau,\dead}$.
\end{snot}

\begin{sdef}
The axiom system of \depACPet\ consists of the equations and conditional
equations presented in Tables~\ref{axioms-pACPet-1} 
and~\ref{axioms-pACPet-2}, on the understanding that $\alpha$ now stands 
for an arbitrary term from $\AProcTermtd$ and $H$ and $I$ now stand for 
an arbitrary subsets of $\AProcTerm$, and in addition the axioms 
presented in Table~\ref{axioms-depACPet}.
\begin{table}[!t]
\caption{Additional axioms of \depACPet}
\label{axioms-depACPet}
\begin{eqntbl}
\begin{axcol}
e = e'         & \mif \Sat{\gD}{\fol{e = e'}}          & \axiom{IMP1} \\
\phi = \psi    & \mif \Sat{\gD}{\fol{\phi \Liff \psi}} & \axiom{IMP2} 
\eqnsep
\True \gc x = x                                      & & \axiom{GC1}  \\
\False \gc x = \dead                                 & & \axiom{GC2}  \\
\phi \gc \dead = \dead                               & & \axiom{GC3}  \\
\phi \gc (x \altc y) = \phi \gc x \altc \phi \gc y   & & \axiom{GC4}  \\
\phi \gc x \seqc y = (\phi \gc x) \seqc y            & & \axiom{GC5}  \\
\phi \gc (\psi \gc x) = (\phi \Land \psi) \gc x      & & \axiom{GC6}  \\
(\phi \Lor \psi) \gc x = \phi \gc x \altc \psi \gc x & & \axiom{GC7}  \\
(\phi \gc x) \leftm y = \phi \gc (x \leftm y)        & & \axiom{GC8}  \\
(\phi \gc x) \commm y = \phi \gc (x \commm y)        & & \axiom{GC9}  \\
x \commm (\phi \gc y) = \phi \gc (x \commm y)        & & \axiom{GC10} \\
\encapa(\phi \gc x) = \phi \gc \encapa(x)            & & \axiom{GC11} \\
\encap{H}(\phi \gc x) = \phi \gc \encap{H}(x)        & & \axiom{GC12} \\
\abstr{I}(\phi \gc x) = \phi \gc \abstr{I}(x)        & & \axiom{GC13} 
\eqnsep
\eval{\sigma}(\ep) = \ep                             & & \axiom{V0}   \\
\eval{\sigma}(\tau \seqc x) = \tau \seqc \eval{\sigma}(x)
                                                     & & \axiom{V1}   \\
\eval{\sigma}(a \seqc x) = a \seqc \eval{\sigma}(x)  & & \axiom{V2}   \\
\eval{\sigma}(a(e) \seqc x) = a(\sigma(e)) \seqc \eval{\sigma}(x)
                                                     & & \axiom{V3}   \\
\eval{\sigma}(\ass{v}{e} \seqc x) = 
{\ass{v}{\sigma(e)} \seqc \eval{\sigma\mapupd{\sigma(e)}{v}}(x)} 
                                                     & & \axiom{V4}   \\
\eval{\sigma}(x \altc y) = \eval{\sigma}(x) \altc \eval{\sigma}(y)
                                                     & & \axiom{V5}   \\
\eval{\sigma}(\phi \gc y) = \sigma(\phi) \gc \eval{\sigma}(x)
                                                     & & \axiom{V6}   
\eqnsep
a(e) \seqc x \commm b(e') \seqc y = (e = e' \gc c(e)) \seqc (x \parc y)
                          & \mif \commf(a,b) = c      & \axiom{CM7Da} \\
a(e) \seqc x \commm b(e') \seqc y = \dead
                          & \mif \commf(a,b) = \dead  & \axiom{CM7Db} \\
a(e) \seqc x \commm \alpha \seqc y = \dead 
                       & \mif \alpha \notin \AProcDPA & \axiom{CM7Dc} \\
\alpha \seqc x \commm a(e) \seqc y = \dead 
                       & \mif \alpha \notin \AProcDPA & \axiom{CM7Dd} \\
\ass{v}{e} \seqc x \commm \alpha \seqc y = \dead    & & \axiom{CM7De} \\
\alpha \seqc x \commm \ass{v}{e} \seqc y = \dead    & & \axiom{CM7Df}   
\eqnsep
\phi \gc (x \paltc{\pi} y) = \phi \gc x \paltc{\pi} \phi \gc y
                                                     & & \axiom{pGC}  \\
\eval{\sigma}(x \paltc{\pi} y) =
 \eval{\sigma}(x) \paltc{\pi} \eval{\sigma}(y)       & & \axiom{pV}   \\
x = x \altc x \Land y = y \altc y \Land 
\encapa(x \altc y) = \dead \Limpl \\ 
\multicolumn{2}{l@{\quad}}{\quad\;\;                                                        
\alpha \seqc 
((\phi \gc \tau \seqc (x \altc y) \altc \phi \gc x) \paltc{\pi} z) = 
\alpha \seqc ((\phi \gc (x \altc y)) \paltc{\pi} z)}  & \axiom{pBED}
\end{axcol}
\end{eqntbl}
\end{table}
In the latter table, 
$\phi$ and $\psi$ stand for arbitrary terms from $\CondTerm$,\,
$e$ and $e'$ stand for arbitrary terms from one of the sets 
$\DataTerm_1$, \ldots, $\DataTerm_n$,\, 
$v$ stands for an arbitrary flexible variable from one of the sets 
$\ProgVar_1$, \ldots, $\ProgVar_n$,\,
$\sigma$ stands for an arbitrary evaluation map from $\EvalMap$,\,
$a,b$, and $c$ stand for arbitrary basic actions from $\Act$, and
$\alpha$ stands for an arbitrary term from $\AProcTermtd$.
Here, ``arbitrary'' must be interpreted as arbitrary to such an extent 
that it does not lead to syntactically incorrect equations and
conditional equations.
\end{sdef}
Axioms GC1--GC10, GC12, and GC13 have been taken from~\cite{BB92c} 
(using a different numbering), but with the axioms with occurrences of 
Hoare's ternary counterpart of the guarded command operator (see below) 
replaced by simpler axioms.
Axioms CM7Da and CM7Db have been inspired by~\cite{BM09d}.
Axiom pBED is axiom BED of \deACPet~\cite{Mid21a} generalized to the 
probabilistic setting in the same way as axiom BE of \ACPet\ has been
generalized to the probabilistic setting in Section~\ref{sect-pACPet}.

\begin{sexa}
The following closed \depACPet\ term of sort $\Proc$ denotes a process
with probabilistic behaviour 
($v_1, \ldots, v_n, v, v', w \in \ProgVar_1$):
\begin{ldispl}
(\PRC{i = 1}{n}{1/n}{\ass{v}{v_i}}) \seqc \ass{v'}{0} \\
{} \seqc
(v_1 = v \gc \ass{v'}{v'+1} \altc \Lnot (v_1 = v) \gc \ep) \\
\hfill \vdots \hfill \\
{} \seqc 
(v_n = v \gc \ass{v'}{v'+1} \altc \Lnot (v_n = v) \gc \ep) \\
{} \seqc
(v' > n/2 \gc \ass{w}{1} \altc
 \Lnot (v' > n/2)  \gc \ass{w}{0})\;.%
\footnotemark
\end{ldispl}%
\footnotetext
{In this example and the next example, the carrier of $\Data_1$ is 
 assumed to be the set of all natural numbers. Moreover, the usual 
 natural number constants, operators on natural numbers, and predicate
 on natural numbers are assumed.}%
The process denoted by this term behaves such that, if there exists a 
natural number $m$ such that $m$ is the initial value of more than $n/2$ 
of the flexible variables $v_1, \ldots, v_n$, then the final value of 
$w$ is $1$ with a probability greater than or equal to $1/2$.
\end{sexa}

\section{\depACPet\ with Recursion}
\label{sect-depACPetr}

A closed \depACPet\ term of sort $\Proc$ denotes a process with a finite 
upper bound to the number of actions that it can perform. 
Recursion allows the description of processes without a finite upper 
bound to the number of actions that it can perform.
In this section, the extension of \depACPet\ with guarded linear 
recursion, called \depACPetr, is presented.

\begin{sdef}
A \emph{recursive specification} over \depACPet\ is a set 
$\set{X_i = t_i \where i \in I}$, where $I$ is a finite set, 
each $X_i$ is a variable from $\cX$, each $t_i$ is a \depACPet\ term of 
sort $\Proc$ in which only variables from $\set{X_i \where i \in I}$ 
occur, and $X_i \neq X_j$ for all $i,j \in I$ with $i \neq j$. 
\end{sdef}
\begin{snot}
\sloppy
We write $\vars(E)$, where $E$ is a recursive specification over 
\depACPet, for the set of all variables that occur in $E$.
\end{snot}
\begin{sdef}
Let $E$ be a recursive specification and let $X \in \vars(E)$.
Then the unique equation $X \!= t \;\in\, E$ is called the 
\emph{recursion equation for $X$ in $E$}.
\end{sdef}

Below, recursive specifications over \depACPet\ are introduced in which 
the right-hand sides of the recursion equations are linear \depACPet\ 
terms.
\begin{sdef}
The set $\LT$ of \emph{linear \depACPet\ terms} is inductively defined 
by the following rules:
\begin{itemize}
\item
$\dead \in \LT$;
\item 
if $\phi \in \CondTerm$, then $\phi \gc \ep \in \LT$;
\item
if $\phi \in \CondTerm$, $\alpha \in \AProcTermt$, and $X \in \cX$, then 
$\phi \gc \alpha \seqc X \in \LT$;
\item
if $t,t' \in \LT \diff \set{\dead}$, then $t \altc t' \in \LT$;
\item
if $t,t' \in \LT$ and $\pi \in \Prob \diff \set{0,1}$, then 
$t \paltc{\pi} t' \in \LT$.
\end{itemize}
Let $t \in \LT$.
Then we refer to the subterms of $t$ that have the form $\phi \gc \ep$ 
or the form $\phi \gc \alpha \seqc X$ as the \emph{summands of} $t$.
\end{sdef}

\begin{sdef}
Let $X$ be a variable from $\cX$ and
let $t$ be an \depACPet\ term in which $X$ occurs. 
Then an occurrence of $X$ in $t$ is \emph{guarded} if $t$ has a subterm 
of the form $\alpha \seqc t'$ where $\alpha \in \AProcTerm$ and $t'$ 
contains this occurrence of $X$.
\end{sdef}
An occurrence of a variable $X$ in a linear \depACPet\ term may not be 
guarded because a linear \depACPet\ term may have summands of the form 
$\phi \gc \tau \seqc X$.

\begin{sdef}
A \emph{guarded linear recursive specification} over \depACPet\ is a 
recursive specification $\set{X_i = t_i \where i \in I}$ over \depACPet\ 
where each $t_i$ is a linear \depACPet\ term, and there does not exist an 
infinite sequence $i_0\;i_1\;\ldots\,$ over $I$ such that, for each 
$k \in \Nat$, there is an occurrence of $X_{i_{k+1}}$ in $t_{i_k}$ 
that is not guarded.
\end{sdef}

\begin{sdef}
A \emph{solution} of a guarded linear recursive specification $E$ over 
\depACPet\ in some model of \depACPet\ is a set 
$\set{p_X \where X \in \vars(E)}$ of elements of the carrier of sort 
$\Proc$ in that model such that each equation in $E$ holds if, for all 
$X \in \vars(E)$, $X$ is assigned~$p_X$. 
\end{sdef}
A guarded linear recursive specification has a unique solution under 
rooted branching bisimulation equivalence as defined in 
Section~\ref{sect-semantics} for \depACPet\ extended with guarded linear 
recursion.
\begin{sdef}
If $\set{p_X \where X \in \vars(E)}$ is the unique solution of a guarded 
linear recursive specification $E$, then, for 
each $X \in \vars(E)$, $p_X$ is called the \emph{$X$-component} of the 
unique solution of $E$. 
\end{sdef}

The algebraic theory \depACPet\ is extended with guarded linear 
recursion by adding constants for solutions of guarded linear recursive 
specifications over \depACPet\ to the signature of \depACPet\ and axioms 
concerning these additional constants to the axiom system of \depACPet.
\begin{sdef}
The signature of the algebraic theory \depACPetr\ consists of the sorts, 
constants, and operators from the signature of \depACPet\ and, for each 
guarded linear recursive specification $E$ over \depACPet\ and 
$X \in \vars(E)$, a constant $\rec{X}{E}$ of sort $\Proc$.
\end{sdef}
\begin{snot}
We write $\ProcTermr$ for the set of all closed $\depACPetr$ terms of 
sort $\Proc$.
\end{snot}
Let $E$ be a guarded linear recursive specification over \depACPet, and 
let $X \in \vars(E)$. 
Then $\rec{X}{E}$ stands for the $X$-component of the unique solution of 
$E$.

\begin{sdef}
\sloppy
The axiom system of the algebraic theory \depACPetr\ consists of the 
equations and conditional equations from the axiom system of \depACPet\
and in addition the equation RDP (Recursive Definition Principle) and 
the conditional equation RSP (Recursive Specification Principle) given 
in Table~\ref{axioms-REC}.
\begin{table}[!t]
\caption{Axioms for guarded linear recursion}
\label{axioms-REC}
\begin{eqntbl}
\begin{axcol}
\rec{X}{E} = \rec{t}{E} & \mif X \!= t \;\in\, E & \axiom{RDP} \\
E \Limpl X = \rec{X}{E} & \mif X \in \vars(E)    & \axiom{RSP} 
\end{axcol}
\end{eqntbl}
\end{table}
In RDP and RSP, $X$ stands for an arbitrary variable from $\cX$, 
\linebreak[2] 
$t$ stands for an arbitrary \depACPet\ term of sort $\Proc$,\, 
$E$ stands for an arbitrary guarded linear recursive specification over 
\depACPet, and 
the notation $\rec{t}{E}$ is used for $t$ with, for all 
$X \in \vars(E)$, all occurrences of $X$ in $t$ replaced by 
$\rec{X}{E}$.
Side conditions restrict what $X$, $t$ and $E$ stand for.
\end{sdef}
RDP and RSP together postulate that guarded linear recursive 
specifications over \depACPet\ have unique solutions: 
the equations $\rec{X}{E} = \rec{t}{E}$ and the conditional equations 
$E \Limpl X \!=\! \rec{X}{E}$ for a fixed $E$ express that the constants 
$\rec{X}{E}$ make up a solution of $E$ and that this solution is the 
only one, respectively.

\sloppy
Because conditional equational formulas must be dealt with in \pACPet, 
\depACPet, and \depACPetr, it is understood that conditional equational 
logic is used in deriving equations from the axioms of these algebraic
theories.
A complete inference system for conditional equational logic can for
example be found in~\cite{BW90,Gog21a}.

\begin{sexa}
\label{example-rec}
The following closed \depACPetr\ term of sort $\Proc$ denotes a process
with probabilistic behaviour ($v \in \ProgVar_1$):
\begin{ldispl}
\rec{X}
 {\set{X = \True \gc \ass{v}{0} \seqc Y,\, 
       Y = \True \gc \ass{v}{v+1} \seqc Y \paltc{1/2} \True \gc \ep}}\;.
\end{ldispl}%
The process denoted by this term behaves such that, for each 
$n \in \Nat$, the probability that $n$ is the final value of $v$ is 
$1/2^{n+1}$.
\end{sexa}

\begin{snot}
We often write $X$ for $\rec{X}{E}$, where $\rec{X}{E} \in \ProcTermr$, 
if $E$ is clear from the context. 
\end{snot}

Let $E$ be a guarded linear recursive specification over \depACPet.
Then there may exist a recursive specification $E'$ over \depACPet\ that 
is not a guarded linear recursive specification such that $E'$ has a 
unique solution and, for all $X \in \vars(E')$, the $X$-component of the 
unique solution of $E'$ is the same as the $X$-component of the unique 
solution of $E$.
In Section~\ref{sect-example}, examples of such recursive specifications
are given.

\begin{scnv}
Let $E$ be a recursive specification over \depACPet\ that has a unique
solution, and let $X \in \vars(E)$.
Then we sloppily say ``$X$ is defined by the recursive specification 
$E$'' instead of ``$X$ is used as a constant standing for the 
$X$-component of the unique solution of $E$''.
\end{scnv}

\section{Branching Bisimulation Semantics}
\label{sect-semantics}

In this section, a structural operational semantics of \depACPetr\ is 
presented and two notions of branching bisimulation equivalence for 
\depACPetr\ based on this structural operational semantics is defined.

The structural operational semantics of \depACPetr\ consists of
\begin{itemize}
\item 
a binary \emph{conditional transition} relation \smash{$\step{\ell}$} on 
$\ProcTermr$ for each $\ell \in \EvalMap \x \AProcTermt$;
\item 
a unary \emph{successful termination} relation $\sterm{\sigma}$ on 
$\ProcTermr$ for each $\sigma \in \EvalMap$;
\item 
a binary \emph{probability} relation $\pstep{}{\ell}{}$ on 
$\ProcTermr$ for each $\ell \in \EvalMap \x \Prob$.
\end{itemize}
\begin{snot}
We write 
\smash{$\astep{t}{\gact{\sigma}{\alpha}}{t'}$} instead of 
\smash{$\tup{t,t'} \in {\step{\tup{\sigma,\alpha}}}$},\;
$\isterm{t}{\sigma}$ instead of $t \in {\sterm{\sigma}}$,\, 
and
\smash{$\pstep{t}{\gprob{\sigma}{\pi}}{t'}$} instead of 
\smash{$\tup{t,t'} \in {\pstep{}{\tup{\sigma,\pi}}{}}$}.
\end{snot}

The relations from the structural operational semantics describe what 
the processes denoted by terms from $\ProcTermr$ are capable of doing as 
follows:
\begin{itemize}
\item
$\astep{t}{\gact{\sigma}{\alpha}}{t'}$: 
if the data values assigned to the flexible variables are as defined by 
$\sigma$, then $t$ has the potential to make a transition to $t'$ by 
performing action~$\alpha$;
\item
$\isterm{t}{\sigma}$: 
if the data values assigned to the flexible variables are as defined  
by~$\sigma$, then $t$ has the potential to terminate successfully;
\item
\smash{$\pstep{t}{\gprob{\sigma}{\pi}}{t'}$}:
if the data values assigned to the flexible variables are as defined by 
$\sigma$, then $t$ has the potential to behave as $t'$ with probability 
$\pi$.
\end{itemize}

\begin{sdef}
The relations belonging to the \emph{structural operational semantics} 
of \depACPetr\ are defined by means of the rules given in 
Tables~\ref{sos-depACPet-1} and~\ref{sos-depACPet-2}.%
\begin{table}[!p]
\caption{Transition rules for \depACPet\ (part~1)}
\label{sos-depACPet-1}
\begin{ruletbl}
{} \\[-3ex]
\Rule
{}
{\astep{\alpha}{\gact{\sigma}{\alpha}}{\ep}}
\\[-2ex]
\Rule
{\phantom{\isterm{\ep}{\sigma}}}
{\isterm{\ep}{\sigma}}
\\
\Rule
{\astep{x}{\gact{\sigma}{\alpha}}{x'},\; \pstep{y}{\gprob{\sigma}{1}}{y'}}
{\astep{x \altc y}{\gact{\sigma}{\alpha}}{x'}}
\quad\;
\Rule
{\pstep{x}{\gprob{\sigma}{1}}{x'},\; \astep{y}{\gact{\sigma}{\alpha}}{y'}}
{\astep{x \altc y}{\gact{\sigma}{\alpha}}{y'}}
\quad\;
\Rule
{\isterm{x}{\sigma},\; \pstep{y}{\gprob{\sigma}{1}}{y'}}
{\isterm{x \altc y}{\sigma}}
\quad\;
\Rule
{\pstep{x}{\gprob{\sigma}{1}}{x'},\; \isterm{y}{\sigma}}
{\isterm{x \altc y}{\sigma}}
\\
\Rule
{\astep{x}{\gact{\sigma}{\alpha}}{x'}}
{\astep{x \seqc y}{\gact{\sigma}{\alpha}}{x' \seqc y}}
\quad\;
\Rule
{\isterm{x}{\sigma},\; \astep{y}{\gact{\sigma}{\alpha}}{y'}}
{\astep{x \seqc y}{\gact{\sigma}{\alpha}}{y'}}
\quad\;
\Rule
{\isterm{x}{\sigma},\; \isterm{y}{\sigma}}
{\isterm{x \seqc y}{\sigma}}
\\
\Rule
{\astep{x}{\gact{\sigma}{\alpha}}{x'},\; \pstep{y}{\gprob{\sigma}{1}}{y'}}
{\astep{x \parc y}{\gact{\sigma}{\alpha}}{x' \parc y}}
\quad\;
\Rule
{\pstep{x}{\gprob{\sigma}{1}}{x'},\; \astep{y}{\gact{\sigma}{\alpha}}{y'}}
{\astep{x \parc y}{\gact{\sigma}{\alpha}}{x \parc y'}}
\quad\;
\RuleC
{\astep{x}{\gact{\sigma}{a}}{x'},\; \astep{y}{\gact{\sigma}{b}}{y'}}
{\astep{x \parc y}{\gact{\sigma}{c}}{x' \parc y'}}
{\commf(a,b) = c}
\\
\RuleC
{\astep{x}{\gact{\sigma}{a(e)}}{x'},\; 
 \astep{y}{\gact{\sigma}{b(e')}}{y'}}
{\astep{x \parc y}{\gact{\sigma}{c(e)}}{x' \parc y'}}
{\commf(a,b) = c,\; \Sat{\gD}{\sigma(\fol{e = e'})}}
\\
\Rule
{\isterm{x}{\sigma},\; \isterm{y}{\sigma}}
{\isterm{x \parc y}{\sigma}}
\\
\Rule
{\astep{x}{\gact{\sigma}{\alpha}}{x'}}
{\astep{x \leftm y}{\gact{\sigma}{\alpha}}{x' \parc y}}
\\
\RuleC
{\astep{x}{\gact{\sigma}{a}}{x'},\; \astep{y}{\gact{\sigma}{b}}{y'}}
{\astep{x \commm y}{\gact{\sigma}{c}}{x' \parc y'}}
{\commf(a,b) = c}
\\
\RuleC
{\astep{x}{\gact{\sigma}{a(e)}}{x'},\; 
 \astep{y}{\gact{\sigma}{b(e')}}{y'}}
{\astep{x \commm y}
  {\gact{\sigma}{c(e)}}{x' \parc y'}}
{\commf(a,b) = c,\; \Sat{\gD}{\sigma(\fol{e = e'})}}
\\
\Rule
{\isterm{x}{\sigma}}
{\isterm{\encapa(x)}{\sigma}}
\\
\RuleC
{\astep{x}{\gact{\sigma}{\alpha}}{x'}}
{\astep{\encap{H}(x)}{\gact{\sigma}{\alpha}}{\encap{H}(x')}}
{\alpha \notin H}
\quad\;
\Rule
{\isterm{x}{\sigma}}
{\isterm{\encap{H}(x)}{\sigma}}
\\
\RuleC
{\astep{x}{\gact{\sigma}{\alpha}}{x'}}
{\astep{\abstr{I}(x)}{\gact{\sigma}{\alpha}}{\abstr{I}(x')}}
{\alpha \notin I}
\quad\;
\RuleC
{\astep{x}{\gact{\sigma}{\alpha}}{x'}}
{\astep{\abstr{I}(x)}{\gact{\sigma}{\tau}}{\abstr{I}(x')}}
{\alpha \in I}
\quad\;
\Rule
{\isterm{x}{\sigma}}
{\isterm{\abstr{I}(x)}{\sigma}}
\\
\RuleC
{\astep{x}{\gact{\sigma}{\alpha}}{x'}}
{\astep{\phi \gc x}{\gact{\sigma}{\alpha}}{x'}}
{\Sat{\gD}{\sigma(\fol{\phi})}}
\quad\;
\RuleC
{\isterm{x}{\sigma}}
{\isterm{\phi \gc x}{\sigma}}
{\Sat{\gD}{\sigma(\fol{\phi})}}
\\
\Rule
{\astep{x}{\gact{\sigma}{\tau}}{x'}}
{\astep{\eval{\sigma}(x)}{\gact{\sigma'}{\tau}}{\eval{\sigma}(x')}}
\quad\;
\Rule
{\astep{x}{\gact{\sigma}{a}}{x'}}
{\astep{\eval{\sigma}(x)}{\gact{\sigma'}{a}}{\eval{\sigma}(x')}}
\quad\;
\Rule
{\astep{x}{\gact{\sigma}{a(e)}}{x'}}
{\astep{\eval{\sigma}(x)}{\gact{\sigma'}{a(\sigma(e))}}
       {\eval{\sigma}(x')}}
\\
\Rule
{\astep{x}{\gact{\sigma}{\ass{v}{e}}}{x'}}
{\astep{\eval{\sigma}(x)}{\gact{\sigma'}{\ass{v}{\sigma(e)}}}
       {\eval{\sigma\mapupd{\sigma(e)}{v}}(x')}}
\quad\;
\Rule
{\isterm{x}{\sigma}}
{\isterm{\eval{\sigma}(x)}{\sigma'}}
\\
\RuleC
{\astep{\rec{t}{E}}{\gact{\sigma}{\alpha}}{x'}}
{\astep{\rec{X}{E}}{\gact{\sigma}{\alpha}}{x'}}
{X \!\!=\! t \,\in\, E}
\quad\;
\RuleC
{\isterm{\rec{t}{E}}{\sigma}}
{\isterm{\rec{X}{E}}{\sigma}}
{X \!\!=\! t \,\in\, E}
\vspace*{1ex}
\end{ruletbl}
\end{table}
\begin{table}[!t]
\caption{Transition rules for \depACPet\ (part~2)}
\label{sos-depACPet-2}
\begin{ruletbl}
{} \\[-3ex]
\Rule
{}
{\pstep{a}{\gprob{\sigma}{1}}{a}}
\qquad
\Rule
{}
{\pstep{\dead}{\gprob{\sigma}{1}}{\dead}}
\qquad
\Rule
{}
{\pstep{\ep}{\gprob{\sigma}{1}}{\ep}}
\\
\Rule
{\pstep{x}{\gprob{\sigma}{\pi}}{x'},\;
 \pstep{y}{\gprob{\sigma}{\rho}}{y'}}
{\pstep{x \altc y}{\gprob{\sigma}{\pi \mul \rho}}{x' \altc y'}}
\\
\Rule
{\pstep{x}{\gprob{\sigma}{\pi}}{x'},\; \pstep{y}{\gprob{\sigma}{1}}{y'}}
{\pstep{x \seqc y}{\gprob{\sigma}{\pi}}{x' \seqc y}}
\qquad
\Rule
{\pstep{x}{\gprob{\sigma}{\pi}}{x'},\; \nisterm{x'}{\sigma}}
{\pstep{x \seqc y}{\gprob{\sigma}{\pi}}{x' \seqc y}}
\qquad
\Rule
{\pstep{x}{\gprob{\sigma}{\pi}}{x'},\; \isterm{x'}{\sigma},\; \pstep{y}{\gprob{\sigma}{\rho}}{y'}}
{\pstep{x \seqc y}{\gprob{\sigma}{\pi \mul \rho}}{x' \seqc y'}}
\\
\Rule
{\pstep{x}{\gprob{\sigma}{\pi}}{x'},\; \pstep{y}{\gprob{\sigma}{\rho}}{y'}}
{\pstep{x \parc y}{\gprob{\sigma}{\pi \mul \rho}}{x' \parc y'}}
\qquad
\Rule
{\pstep{x}{\gprob{\sigma}{\pi}}{x'},\; \pstep{y}{\gprob{\sigma}{\rho}}{y'}}
{\pstep{x \leftm y}{\gprob{\sigma}{\pi \mul \rho}}{x' \leftm y'}}
\qquad
\Rule
{\pstep{x}{\gprob{\sigma}{\pi}}{x'},\; \pstep{y}{\gprob{\sigma}{\rho}}{y'}}
{\pstep{x \commm y}{\gprob{\sigma}{\pi \mul \rho}}{x' \commm y'}}
\\
\Rule
{\pstep{x}{\gprob{\sigma}{\pi}}{x'}}
{\pstep{\encapa(x)}{\gprob{\sigma}{\pi}}{\encapa(x')}}
\qquad
\Rule
{\pstep{x}{\gprob{\sigma}{\pi}}{x'}}
{\pstep{\encap{H}(x)}{\gprob{\sigma}{\pi}}{\encap{H}(x')}}
\qquad
\Rule
{\pstep{x}{\gprob{\sigma}{\pi}}{x'}}
{\pstep{\abstr{I}(x)}{\gprob{\sigma}{\pi}}{\abstr{I}(x')}}
\\
\Rule
{\pstep{x}{\gprob{\sigma}{\rho}}{z},\; \pstep{y}{\gprob{\sigma}{\rho'}}{z}}
{\pstep{x \paltc{\pi} y}{\gprob{\sigma}{\pi \mul \rho + (1-\pi) \mul \rho'}}{z}}
\\
\RuleC
{\pstep{x}{\gprob{\sigma}{\pi}}{x'}}
{\pstep{\phi \gc x}{\gprob{\sigma}{\pi}}{\phi \gc x'}}
{\Sat{\gD}{\sigma(\fol{\phi})}}
\qquad
\RuleC
{\pstep{x}{\gprob{\sigma}{\pi}}{x'}}
{\pstep{\phi \gc x}{\gprob{\sigma}{1}}{\phi \gc x}}
{\nSat{\gD}{\sigma(\fol{\phi})}}
\\
\Rule
{\pstep{x}{\gprob{\sigma}{\pi}}{x'}}
{\pstep{\eval{\sigma}(x)}{\gprob{\sigma'}{\pi}}{\eval{\sigma}(x')}}
\\
\RuleC
{\pstep{\srec{t}{E}}{\gprob{\sigma}{\pi}}{z}}
{\pstep{\rec{X}{E}}{\gprob{\sigma}{\pi}}{z}}
{X = t\; \in \;E}
\\
\Rule
{\npstep{x}{\gprob{\sigma}{\pi}}{x'}
 \;\text{for all}\; \pi \in \Prob \diff \set{0}}
{\pstep{x}{\gprob{\sigma}{0}}{x'}}
\vspace*{1ex}
\end{ruletbl}
\end{table}
In these tables, 
$\sigma$ and $\sigma'$ stand for arbitrary evaluation maps from 
$\EvalMap$,\,
$\alpha$ stands for an arbitrary action from~$\AProcTermt$,\,
$a,b$, and $c$ stand for arbitrary actions from~$\Act$,\,
$e$ and $e'$ stand for arbitrary terms from one of the sets 
$\DataTerm_1$, \ldots, $\DataTerm_n$,\,
$H$ and $I$ stands for arbitrary subsets of~$\AProcTerm$,\,
$\phi$~stands for an arbitrary term from $\CondTerm$,\,
$v$ stands for an arbitrary flexible variable from one of the sets 
$\ProgVar_1$, \ldots, $\ProgVar_n$,\,
$X$ stands for an arbitrary variable from $\cX$,\, 
$t$ stands for an arbitrary \depACPet\ term of sort $\Proc$, 
$E$ stands for an arbitrary guarded linear recursive specification over 
\depACPet, and
$\pi$, $\rho$, and $\rho'$ stand for arbitrary probabilities from 
$\Prob$.
Here, ``arbitrary'' must be interpreted as arbitrary to such an extent 
that it does not lead to syntactically incorrect terms.
It should be clear that $\nisterm{t}{\sigma}$ stands for the negation of 
$\isterm{t}{\sigma}$ and \smash{$\npstep{t}{\gprob{\sigma}{\pi}}{t'}$} 
stands for the negation of \smash{$\pstep{t}{\gprob{\sigma}{\pi}}{t'}$}.
\end{sdef}

Below, a premise of the form $\nisterm{t}{\sigma}$ or 
\smash{$\npstep{t}{\gprob{\sigma}{\pi}}{t'}$} is called a negative 
premise,
$\isterm{t}{\sigma}$ is called the denial of $\nisterm{t}{\sigma}$, and
\smash{$\pstep{t}{\gprob{\sigma}{\pi}}{t'}$} is called the denial of
\smash{$\npstep{t}{\gprob{\sigma}{\pi}}{t'}$}.

Because of the presence of negative premises, some explanation is needed
about how the relations from the structural operational semantics of 
\depACPetr\ are defined by means of the rules given in 
Tables~\ref{sos-depACPet-1} and~\ref{sos-depACPet-2}. 
Because these rules constitute a well-supported complete transition 
system specification (see e.g.~\cite{Gla04a}),
\smash{$\astep{t}{\gact{\sigma}{\alpha}}{t'}$}, $\isterm{t}{\sigma}$ or 
\smash{$\pstep{t}{\gprob{\sigma}{\pi}}{t'}$}, holds iff it is provable 
from the rules given in Tables~\ref{sos-depACPet-1} 
and~\ref{sos-depACPet-2}, where $\nisterm{t}{\sigma}$ or 
\smash{$\npstep{t}{\gprob{\sigma}{\pi}}{t'}$} is considered provable if 
every conceivable proof of its denial involves a negative premise of 
which the denial has already been proved. 

We could have excluded the relations 
\smash{$\pstep{}{\gprob{\sigma}{0}}{}$} and by that have obviated the 
need for the last rule in Table~\ref{sos-depACPet-2}.
In that case, however, 11 additional rules concerning the relations 
\smash{$\pstep{}{\gprob{\sigma}{\pi}}{}$}, all with negative premises, 
would be needed instead.

Notice that, if $t$ is not derivably equal to a term whose outermost 
operator is a probabilistic choice operator, then $t$ can only behave as 
itself and consequently we have that 
\smash{$\pstep{t}{\gprob{\sigma}{1}}{t}$} and 
\smash{$\pstep{t}{\gprob{\sigma}{0}}{t'}$} for each term $t'$ other than 
$t$.

The next two propositions express properties of the relations 
\smash{$\pstep{}{\gprob{\sigma}{\pi}}{}$}.
\begin{proposition}
\label{prop-pstep-1}
For all $\sigma \in \EvalMap$ and $t,t' \in \ProcTermr$, 
$\pstep{t}{\gprob{\sigma}{1}}{t'}$ only if $t \equiv t'$.
\end{proposition}
\begin{proof}
This is easy to prove by induction on the structure of $t$.
\qed
\end{proof}
\begin{proposition}
\label{prop-pstep-2}
For all $\sigma \in \EvalMap$ and $t,t' \in \ProcTermr$, there exists a 
$\pi \in \Prob$ such that $\pstep{t}{\gprob{\sigma}{\pi}}{t'}$.
\end{proposition}
\begin{proof}
This is easy to prove by induction on the structure of $t$.
\qed
\end{proof}

The probability relations give rise to probability distribution 
functions.
\begin{sdef}
For each $\sigma \in \EvalMap$, we define a probability distribution 
function $P_\sigma$ from $\ProcTermr \x \ProcTermr$ to $\Prob$ as 
follows:
\begin{ldispl}
P_\sigma(t,t') = 
{\displaystyle \sum_{\pi \in \Pi_\sigma(t,t')}} \!\pi \;,
\quad\mathrm{where}\;\; 
\Pi_\sigma(t,t') = 
\set{\pi \where \pstep{t}{\gprob{\sigma}{\pi}}{t'}}\;.
\end{ldispl}%
\end{sdef}
The function $P_\sigma$ can be explained as follows:
$P_\sigma(t,t')$ is the total probability that $t$ will behave as $t'$ 
if the data values assigned to the flexible variables are as defined by 
$\sigma$.
\begin{snot}
We write $P_\sigma(t,T)$, where $\sigma \in \EvalMap$, 
$t \in \ProcTermr$, and $T \subseteq \ProcTermr$, 
for $\sum_{t' \in T} P_\sigma(t,t')$.
\end{snot}

The well-definedness of $P_\sigma$ is a corollary of 
Proposition~\ref{prop-pstep-2}.
\begin{corollary}
\label{corol-well-def-pdf}
For all $\sigma \in \EvalMap$ and $t,t' \in \ProcTermr$, there exists 
a unique $\pi \in \Prob$ such that $P_\sigma(t,t') = \pi$.
\end{corollary}

$P_\sigma$ is a probability distribution function indeed.
\begin{proposition}
\label{prop-pdf}
For all $\sigma \in \EvalMap$ and $t \in \ProcTermr$, 
$P_\sigma(t,\ProcTermr) = 1$.
\end{proposition}
\begin{proof}
This is easy to prove by induction on the structure of $t$.
\qed
\end{proof}

It follows from Propositions~\ref{prop-pstep-1} and~\ref{prop-pdf} that 
the behaviour of $t$ does not start with a probabilistic choice if 
\smash{$\pstep{t}{\gprob{\sigma}{1}}{t'}$}.
This explains the premises \smash{$\pstep{x}{\gprob{\sigma}{1}}{x'}$} 
and \smash{$\pstep{y}{\gprob{\sigma}{1}}{y'}$} in
Table~\ref{sos-depACPet-1}: they guarantee that probabilistic choices 
are always resolved before choices involved in alternative composition 
and parallel composition are resolved.

\begin{snot}
We write $[t]_R$, where $t \in \ProcTermr$ and $R$ is an equivalence 
relation on $\ProcTermr$, for the equivalence class of $t$ with respect 
to $R$.
\end{snot}

Two processes are considered equal if they can simulate each other 
insofar as their observable behavioural potentials are concerned, taking 
into account the assigments of data values to flexible variables under 
which the potentials are available.
This can be dealt with by means of a variant of the notion of branching 
bisimulation equivalence introduced in~\cite{GW96a} that is adapted to 
the current setting.

An equivalence relation on the set $\AProcTermt$ is needed.
\begin{sdef}
Two actions $\alpha,\alpha' \in \AProcTermt$ are \emph{data equivalent}, 
written $\alpha \simeq \alpha'$, iff one of the following holds:
\begin{itemize}
\item
there exists an $a \in \Actt$ such that $\alpha = a$ 
and $\alpha' = a$;
\item
there exist an $a \in \Act$ and $e,e' \in \DataTerm_i$ for some
$i \in \set{1,\ldots,n}$ such that $\Sat{\gD}{\fol{e = e'}}$, 
$\alpha = a(e)$, and $\alpha' = a(e')$;
\item
there exist a $v \in \ProgVar_i$ and $e,e' \in \DataTerm_i$ for some
$i \in \set{1,\ldots,n}$ such that \mbox{$\Sat{\gD}{\fol{e = e'}}$}, 
$\alpha = \ass{v}{e}$, and $\alpha' = \ass{v}{e'}$.
\end{itemize}
\end{sdef}
\begin{snot}
We write $[\alpha]$, where $\alpha \in \AProcTermt$, for the equivalence 
class of $\alpha$ with respect to $\simeq$.
\end{snot}

\begin{snot}
We write \smash{$t \sord{\sigma} t'$} to indicate that either 
\smash{$\astep{t}{\gact{\sigma}{\tau}}{t'}$} or there exists a 
$\pi \in \Prob$ such that \smash{$\pstep{t}{\gprob{\sigma}{\pi}}{t'}$}.
\end{snot}

\begin{sdef}
A \emph{branching bisimulation} is an equivalence relation $R$ on 
$\ProcTermr$ such that, for all terms $t_1,t_2 \in \ProcTermr$ with 
$(t_1,t_2) \in R$, the following \emph{transfer conditions} hold:
\begin{itemize}
\item
if $\astep{t_1}{\gact{\sigma}{\alpha}}{t_1'}$, then there exist  
${t_2}_1,\ldots,{t_2}_n,t_2' \in \ProcTermr$ ($n \in \Nat$) and an 
$\alpha' \in [\alpha]$ such that:
\begin{itemize}
\item
$t_2 \sord{\sigma} {t_2}_1 \sord{\sigma} \,\cdots\,
 \sord{\sigma} {t_2}_n$;
\item
either \smash{$\astep{{t_2}_n}{\gact{\sigma}{\alpha'}}{t_2'}$} 
or $\alpha \equiv \tau$, $t_2 \equiv t_2'$, and 
\smash{$\nisterm{t_2'}{\sigma}$};
\item
$(t_1,{t_2}_i) \in R$ for all $i \in \set{1,\ldots,n}$ and
$(t_1',t_2') \in R$;
\end{itemize}
\item
if $\isterm{t_1}{\sigma}$, then $\isterm{t_2}{\sigma}$;
\item
$P_\sigma(t_1,[t]_R) = P_\sigma(t_2,[t]_R)$ for all 
$\sigma \in \EvalMap$ and $t\in \ProcTermr$.
\end{itemize}
Two terms $t_1,t_2 \in \ProcTermr$ are 
\emph{branching bisimulation equivalent}, written 
$t_1 \bbisim t_2$, if there exists a branching bisimulation $R$ 
such that $(t_1,t_2) \in R$.
\end{sdef}

\begin{proposition}
\label{prop-bbisim-largest}
The relation $\bbisim$ is the largest branching bisimulation.
\end{proposition}
\begin{proof}
In the same way as Proposition~3.2.4 in~\cite{Geo11a}, it can be shown 
that the transitive closure of the union over a set of branching 
bisimulations is a branching bisimulation.
Moreover, $\bbisim$ is the union over the set of all branching 
bisimulations according to the definition of $\bbisim$.
Hence, $\bbisim$ is the largest branching bisimulation. 
\qed
\end{proof}
The following is a corollary of Proposition~\ref{prop-bbisim-largest} 
and the fact that a branching bisimulation is an equivalence relation by 
definition.
\begin{corollary}
\label{corollary-bbisim-equiv}
The relation $\bbisim$ is an equivalence relation.
\end{corollary}

\begin{sdef}
A \emph{rooted branching bisimulation} is a binary relation $R$ on 
$\ProcTermr$ such that, for all terms $t_1,t_2 \in \ProcTermr$ with 
$(t_1,t_2) \in R$, the following transfer conditions hold:
\begin{itemize}
\item
if $\astep{t_1}{\gact{\sigma}{\alpha}}{t_1'}$, then there exist an 
$\alpha' \in [\alpha]$ and a $t_2' \in \ProcTermr$ such that 
\smash{$\astep{t_2}{\gact{\sigma}{\alpha'}}{t_2'}$} and 
$t_1' \bbisim t_2'$;
\item
if $\astep{t_2}{\gact{\sigma}{\alpha}}{t_2'}$, then there exist an 
$\alpha' \in [\alpha]$ and a $t_1' \in \ProcTermr$ such that 
\smash{$\astep{t_1}{\gact{\sigma}{\alpha'}}{t_1'}$} and 
$t_1' \bbisim t_2'$;
\item
$\isterm{t_1}{\sigma}$ iff $\isterm{t_2}{\sigma}$;
\item
$P_\sigma(t_1,[t]_R) = P_\sigma(t_2,[t]_R)$ for all 
$\sigma \in \EvalMap$ and $t\in \ProcTermr$.
\end{itemize}
Two terms $t_1,t_2 \in \ProcTermr$ are 
\emph{rooted branching bisimulation equivalent}, written 
$t_1 \rbbisim t_2$, if there exists a rooted branching bisimulation $R$ 
such that $(t_1,t_2) \in R$.

\noindent
Let $R$ be a rooted branching bisimulation such that $(t_1,t_2) \in R$.
Then we say that $R$ is a rooted branching bisimulation 
\emph{witnessing} $t_1 \rbbisim t_2$.
\end{sdef}

\begin{proposition}
\label{prop-rbbisim-largest}
The relation $\rbbisim$ is the largest rooted branching bisimulation.
\end{proposition}
\begin{proof}
This is proved in almost the same way as 
Proposition~\ref{prop-bbisim-largest}.
\qed
\end{proof}
\begin{proposition}
\label{proposition-rbbisim-equiv}
The relation $\rbbisim$ is an equivalence relation.
\end{proposition}
\begin{proof}
The following are corollaries of Corollary~\ref{corollary-bbisim-equiv} 
and the definition of a rooted branching bisimulation:
\begin{itemize}
\item
for all $t \in \ProcTermr$,
the identity relation $I$ on $\ProcTermr$ is a rooted branching 
bisimulation witnessing $t \rbbisim t$;
\item
for all $t_1, t_2 \in \ProcTermr$ such that $t_1 \rbbisim t_2$,
if $R$ is a rooted branching bisimulation witnessing $t_1 \rbbisim t_2$,
then $R^{-1}$ is a rooted branching bisimulation witnessing 
$t_2 \rbbisim t_1$;
\item
for all $t_1, t_2, t_3 \in \ProcTermr$ such that $t_1 \rbbisim t_2$ and
$t_2 \rbbisim t_3$, 
if $R$ is a rooted branching bisimulation witnessing $t_1 \rbbisim t_2$ 
and $R'$ is a rooted branching bisimulation witnessing 
$t_2 \rbbisim t_3$,
then $R \circ R'$ is a rooted branching bisimulation witnessing 
$t_1 \rbbisim t_3$.%
\footnote{We write $R \circ R'$ for the composition of $R$ with $R'$.}
\end{itemize}
This means that $\rbbisim$ is reflexive, symmetric, and transitive.
Hence, $\rbbisim$ is an equivalence relation.
\qed
\end{proof}

In Section~\ref{sect-sound-compl}, it is proved that $\rbbisim$ is a 
congruence with respect to the operators of \depACPetr\ of which the 
result sort and at least one argument sort is~$\Proc$.
In fact, $\bbisim$ is also a congruence with respect to these operators, 
except the operator $\altc$\,.

\begin{sexa}
Let $t,t',s,s' \in \ProcTermr$, $\sigma \in \EvalMap$, and
$\alpha, \beta \in \AProcTerm$.
Then we have:
\begin{ldispl}
\mathrm{if}\; \astep{t}{\gprob{\sigma}{\alpha}}{t'} \;\mathrm{then}\;
\tau \seqc t \bbisim t \;.
\end{ldispl}%
However, we also have:
\begin{ldispl}
\mathrm{if}\; \astep{t}{\gprob{\sigma}{\alpha}}{t'} 
\;\mathrm{and}\; \astep{s}{\gprob{\sigma}{\beta}}{s'} \;\mathrm{then}\;
\tau \seqc t \altc s \not\bbisim t \altc s\;.
\end{ldispl}%
So we have that $\bbisim$ is not a congruence with respect to the 
operator $\altc$.
In order to make $\rbbisim$ a congruence with respect to the operator 
$\altc$, it is defined such that 
\begin{ldispl}
\mathrm{if}\; \astep{t}{\gprob{\sigma}{\alpha}}{t'} \;\mathrm{then}\;
\tau \seqc t \not\rbbisim t \;.
\end{ldispl}%
\end{sexa}

The third transfer condition given in the definition of a branching 
bisimulation is a rather strong condition.
\begin{sexa}
Let $t,t',s,s' \in \ProcTermr$, $\sigma \in \EvalMap$, 
$\pi \in \Prob \diff \set{0,1}$, and $\alpha \in \AProcTermt$.
Then we have for all equivalence relations $R$ on $\ProcTermr$:
\begin{ldispl}
\mathrm{if}\; \pstep{t}{\gprob{\sigma}{\pi}}{t'} \;\mathrm{then}\;
P_\sigma(\tau \seqc t,[t']_R) = 0 \;\mathrm{and}\; 
P_\sigma(t,[t']_R) > 0\;.
\end{ldispl}%
So we have $\tau \seqc t \not\bbisim t$ and 
$\alpha \seqc (\tau \seqc t) \not\rbbisim \alpha \seqc t$
if $\pstep{t}{\gprob{\sigma}{\pi}}{t'}$.
The reason for this is that otherwise $\bbisim$ and $\rbbisim$ would not 
be a congruence with respect to the operator $\parc$ because we have for 
all equivalence relations $R$ on $\ProcTermr$:
\begin{ldispl}
\mathrm{if}\; \pstep{t}{\gprob{\sigma}{\pi}}{t'} 
\;\mathrm{and}\; \astep{s}{\gprob{\sigma}{\alpha}}{s'} 
\\ \qquad
\;\mathrm{then}\;
P_\sigma(\tau \seqc t \parc s,[t' \parc s]_R) = 0 \;\mathrm{and}\; 
P_\sigma(t \parc s,[t' \parc s]_R) > 0\;.
\end{ldispl}%
\end{sexa}

Without the occurrence of $\nisterm{t_1'}{\sigma}$ in the first transfer 
condition given in the definition of a branching bisimulation, 
$\bbisim$ and $\rbbisim$ would not be congruences with respect to the 
operator $\seqc$\,.
\begin{sexa}
Let $t,t' \in \ProcTermr$, $\sigma \in \EvalMap$,  
$\pi \in \Prob \diff \set{0,1}$, and $\alpha \in \AProcTermt$.
Then, without the occurrence of $\nisterm{t_1'}{\sigma}$ in the first 
transfer condition given in the definition of a branching bisimulation, 
we would have $\tau \seqc \ep \bbisim \ep$ and 
$\alpha \seqc (\tau \seqc \ep) \rbbisim \alpha \seqc \ep$.
However, we have for all equivalence relations $R$ on $\ProcTermr$:
\begin{ldispl}
\mathrm{if}\; \pstep{t}{\gprob{\sigma}{\pi}}{t'} \;\mathrm{then}\;
P_\sigma((\tau \seqc \ep) \seqc t,[t']_R) = 0
\;\mathrm{and}\;
P_\sigma(\ep \seqc t,[t']_R) > 0\;.
\end{ldispl}%
So we would have $(\tau \seqc \ep) \seqc t \not\bbisim \ep \seqc t$ and 
$\alpha \seqc ((\tau \seqc \ep) \seqc t) \not\rbbisim
 \alpha \seqc (\ep \seqc t)$
if \smash{$\pstep{t}{\gprob{\sigma}{\pi}}{t'}$}.
This means that $\bbisim$ and $\rbbisim$ would not be congruences with 
respect to the operator~$\seqc$\,.
\end{sexa}

\section{Soundness and Completeness of the Axiom System}
\label{sect-sound-compl}

This section concerns the issue of soundness and completeness of the 
axiom system of \depACPetr\ with respect to branching bisimulation 
equivalence.
It will be shown that the axiom system is sound and it will be explained
why the axiom system is incomplete.

\pagebreak[2]

Rooted branching bisimulation equivalence is an equivalence 
relation.
Moreover, rooted branching bisimulation equivalence is a congruence with 
respect to the operators of \depACPetr\ of which the result sort and at 
least one argument sort is $\Proc$.
\begin{proposition}[Congruence]
\label{proposition-congr-deACPet}
For all terms $t_1,t_1',t_2,t_2' \in \ProcTermr$ and all terms 
$\phi \in \CondTerm$, 
$t_1 \rbbisim t_2$ and $t_1' \rbbisim t_2'$ only if 
$t_1 \altc t_1' \rbbisim t_2 \altc t_2'$, 
$t_1 \seqc t_1' \rbbisim t_2 \seqc t_2'$, 
$t_1 \parc t_1' \rbbisim t_2 \parc t_2'$, 
$t_1 \leftm t_1' \rbbisim t_2 \leftm t_2'$,
$t_1 \commm t_1' \rbbisim t_2 \commm t_2'$,
$\encapa(t_1) \rbbisim \encapa(t_2)$, 
$\encap{H}(t_1) \rbbisim \encap{H}(t_2)$, 
$\abstr{I}(t_1) \rbbisim \abstr{I}(t_2)$,
$t_1 \paltc{\pi} t_1' \rbbisim t_2 \paltc{\pi} t_2'$, 
$\phi \gc t_1 \rbbisim \phi \gc t_2$, and 
$\eval{\sigma}(t_1) \rbbisim \eval{\sigma}(t_2)$.
\end{proposition}
\begin{proof}
Below, we write $R_1 \diamond R_2$, where $R_1$ and $R_2$ are 
rooted branching bisimulations and 
$\diamond \in
 \set{\altc,\seqc,\parc,\leftm,\commm} \union
 \set{\paltc{\pi} \where \pi \in \Prob}$, 
for the relation
$\set{\tup{t_1 \diamond t_2,t_1' \diamond t_2'} \where
      R_1(t_1,t_1') \Land R_2(t_2,t_2')}$
and
$\diamond(R)$, where $R$ is a rooted branching bisimulation and
$\diamond \in
 \set{\encapa} \union
 \set{\encap{H} \where H \subseteq \AProcTerm} \union
 \set{\abstr{I} \where I \subseteq \AProcTerm} \union
 \set{\eval{\sigma} \where \sigma \in \EvalMap}$,
for the relation
$\set{\tup{\diamond(t_1),\diamond(t_1')} \where R(t_1,t_1')}$.

Let $t_1,t_1',t_2,t_2' \in \ProcTermr$ be such that $t_1 \rbbisim t_1'$ 
and $t_2 \rbbisim t_2'$, and 
let $R_1$ and $R_2$ be rooted branching bisimulations witnessing 
$t_1 \rbbisim t_1'$ and $t_2 \rbbisim t_2'$, respectively.

For each operator 
$\diamond \in
 \set{\altc,\seqc,\parc,\leftm,\commm} \union
 \set{\paltc{\pi} \where \pi \in \Prob}$, 
we construct an equivalence relation $R_\diamond$ on $\ProcTermr$ as 
follows ($\pi \in \Prob$):
\begin{ldispl}
\begin{ceqns}
\text{in the case that } 
\diamond \text{ is } \seqc : &
R_\diamond = ((R_1 \diamond R_2) \union R_2)\eqvcl\;; 
\\
\text{in the case that } 
\diamond \text{ is } \altc, \parc \text{ or } \paltc{\pi} : &
R_\diamond = ((R_1 \diamond R_2) \union R_1 \union R_2)\eqvcl\;; 
\\
\text{in the case that } 
\diamond \text{ is } \leftm \text{ or } \commm : &
R_\diamond = ((R_1 \diamond R_2) \union (R_1 \parc R_2) \union
     R_1 \union R_2)\eqvcl\;\;
\end{ceqns}
\end{ldispl}%
and for each operator 
$\diamond \in
 \set{\encapa} \union
 \set{\encap{H} \where H \subseteq \AProcTerm} \union
 \set{\abstr{I} \where I \subseteq \AProcTerm} \union
 \set{\eval{\sigma} \where \sigma \in \EvalMap}$, 
we construct an equivalence relation $R_\diamond$ on $\ProcTermr$ as 
follows ($H,I \subseteq \AProcTerm$): 
\begin{ldispl}
\begin{ceqns}
\text{in the case that } 
\diamond \text{ is } \encapa, \encap{H} \text{ or }\abstr{I} : &
R_\diamond = \diamond(R_1)\;;
\\
\text{in the case that } 
\diamond \text{ is } \eval{\sigma} : &
R_\diamond = 
(\Union_{\sigma' \in \EvalMap} \eval{\sigma'}(R_1))\eqvcl\;.
\end{ceqns}
\end{ldispl}%
Moreover, for each term $\phi \in \CondTerm$, we construct an 
equivalence relation $R_{\phi \gc {}}$ on $\ProcTermr$ as 
follows: 
\begin{ldispl}
R_{\phi \gc {}} = 
(\set{\tup{\phi \gc t_1,\phi \gc t_1'}} \union R_1)\eqvcl\;.
\end{ldispl}%

For each of the constructed equivalence relations, we have to show that 
the transfer conditions from the definition of a rooted branching 
bisimulation hold.

The proofs that the conditions concerning the relations 
$\step{\gact{\sigma}{\alpha}}$ and $\sterm{\sigma}$ hold are easy.
The proof that the condition concerning the functions $P_\sigma$ holds 
is straightforward using the following easy-to-check properties of 
$P_\sigma$ ($\sigma \in \EvalMap$):
\begin{ldispl}
\begin{tabular}{@{}l@{}}
if $I$ is an index set, 
\\ \phantom{if}
for all $i \in I$, $T_i \subseteq \ProcTermr$ and, 
\\ \phantom{if}
for all $i,j \in I$ with $i \neq j$, $T_i \inter T_j = \emptyset$, 
\\
then $P_\sigma(t,\Union_{i \in I} T_i) = \sum_{i \in I} P_\sigma(t,T_i)$
\end{tabular}
\end{ldispl}%
and
\begin{ldispl}
\begin{aceqns}
P_\sigma(t \seqc t',T \seqc T')     & = & 0      & 
 \text{if } t' \notin T'\;, \\
P_\sigma(t \seqc t',T \seqc T')     & = & P_\sigma(t,T) & 
 \text{if } t' \in T' \text{ and } \nisterm{t}{\sigma}\;,  \\
P_\sigma(t \seqc t',T \seqc T')     & = &
 P_\sigma(t,T) \mul P_\sigma(t',T') & 
 \text{if } t' \in T' \text{ and }\, \isterm{t}{\sigma}\;, \\
P_\sigma(t \altc t',T \altc T')     & = &
 P_\sigma(t,T) \mul P_\sigma(t',T')\;, \\
P_\sigma(t \parc t',T \parc T')     & = &
 P_\sigma(t,T) \mul P_\sigma(t',T')\;, \\
P_\sigma(t \leftm t',T \leftm T')   & = &
 P_\sigma(t,T) \mul P_\sigma(t',T')\;, \\
P_\sigma(t \commm t',T \commm T')   & = &
 P_\sigma(t,T) \mul P_\sigma(t',T')\;, \\
P_\sigma(\encapa(t),\encapa(T))     & = & P_\sigma(t,T)\;, \\
P_\sigma(\encap{H}(t),\encap{H}(T)) & = & P_\sigma(t,T)\;, \\
P_\sigma(\abstr{I}(t),\abstr{I}(T)) & = & P_\sigma(t,T)\;, \\
P_\sigma(t \paltc{\pi} t',T)        & = &
 \pi \mul P_\sigma(t,T) + (1-\pi) \mul P_\sigma(t',T)\;,   \\
P_\sigma(\phi \gc t),\phi \gc T)    & = & P_\sigma(t,T) &
 \text{if } \Sat{\gD}{\sigma(\fol{\phi})}\;, \\
P_\sigma(\phi \gc t),\phi \gc T)    & = & 1 &
 \text{if } \nSat{\gD}{\sigma(\fol{\phi})}\;, \\
P_\sigma(\eval{\sigma'}(t),\eval{\sigma'}(T)) 
                                 & = & P_{\sigma'}(t,T)\;.
\end{aceqns}
\end{ldispl}%
where we write 
$T \diamond T'$, where $T,T' \subseteq \ProcTermr$ and 
$\diamond \in
 \set{\altc,\seqc,\parc,\leftm,\commm} \union
 \set{\paltc{\pi} \where \pi \in \Prob}$, 
for the set $\set{t \diamond t' \where t \in T \Land t' \in T'}$; 
$\diamond(T)$, where $T \subseteq \ProcTermr$ and
$\diamond \in
 \set{\encapa} \union \linebreak[2]
 \set{\encap{H} \where H \subseteq \AProcTerm} \union
 \set{\abstr{I} \where I \subseteq \AProcTerm} \union
 \set{\eval{\sigma} \where \sigma \in \EvalMap}$,
for the set $\set{\diamond(t) \where t \in T}$; and
$\phi \gc T$, where $T \subseteq \ProcTermr$ and $\phi \in \CondTerm$,
for the set $\set{\phi \gc t \where t \in T}$.
\qed
\end{proof}

Below, the following definition of validity will be used.
\begin{sdef}
An equation $\eqn$ of \depACPetr\ terms of sort $\Proc$ is said to be 
\emph{valid with respect to} ${\rbbisim}$ if, for each closed 
substitution instance $t = t'$ of $\eqn$, $t \rbbisim t'$.
A conditional equation $\ceqn$ of \depACPetr\ terms of sort $\Proc$ 
is said to be \emph{valid with respect to} ${\rbbisim}$ if, for each 
closed substitution instance 
$\set{t_i = t'_i \where i \in I} \Limpl t = t'$ 
of $\ceqn$, $t \rbbisim t'$ if $t_i \rbbisim t'_i$ for each $i \in I$.
\end{sdef}

The axiom system of \depACPetr\ is sound with respect to ${\rbbisim}$ for 
equations between terms from $\ProcTermr$.
\begin{theorem}[Soundness]
\label{theorem-soundness-ACPet}
For all terms $t,t' \in \ProcTermr$, $t = t'$ is derivable from the 
axioms of \depACPetr\ only if $t \rbbisim t'$.
\end{theorem}
\begin{proof}
We write $\csi(\eqn)$, where $\eqn$ is an equation between \depACPetr\ 
terms of sort $\Proc$, for the binary relation on $\ProcTermr$ that 
consists of all tuples $\tup{t,t'}$ such that $t = t'$ is a closed 
substitution instances of $\eqn$.
Moreover, we write $R\eqvcl$, where $R$ is a binary relation, for the 
equivalence closure of $R$.

Because ${\rbbisim}$ is a congruence with respect to all operators from 
the signature of \depACPetr, only the validity of each axiom of 
\depACPetr\ with respect to ${\rbbisim}$ has to be proved.

For each axiom $\ax$ of \depACPetr, a rooted branching bisimulation 
$R_\ax$ witnessing the validity of $\ax$ with respect to ${\rbbisim}$ 
can be constructed as follows:
\begin{itemize}
\item
if $\ax$ is one of the axioms that is an equational axiom:
\begin{ldispl}
R_\ax = \csi(\ax)\eqvcl\;;
\end{ldispl}%
\item
if $\ax$ is CM1E$'$:
\begin{ldispl}
R_\ax = 
(\set{\tup{t_1 \parc t_2,
      t_1 \leftm t_2 \altc t_2 \leftm t_1 \altc t_1 \commm t_2 \altc
      \encapa(t_1) \seqc \encapa(t_2)} \where {}
\\ \phantom{R_\ax = (\{}
      t_1,t_2 \in \ProcTermr \Land
      t_1 \rbbisim t_1 \altc t_1 \Land t_2 \rbbisim t_2 \altc t_2}
\\ \phantom{R_\ax = (}\,
 {} \union 
 \{\tup{t_1 \parc t_2, t_2 \parc t_1} \where {}
\\ \phantom{R_\ax = (\, {} \union \{}
   t_1,t_2 \in \ProcTermr \Land
   t_1 \rbbisim t_1 \altc t_1 \Land t_2 \rbbisim t_2 \altc t_2\})\eqvcl\;;
\end{ldispl}%
\item
if $\ax$ is an instance of pBE:
\begin{ldispl}
R_\ax \\ \, {} = 
(\{\tup{\alpha \seqc
        ((\tau \seqc (t_1 \altc t_2) \altc t_1) \paltc{\pi} t_3),
        \alpha \seqc (t_1 \altc t_2) \paltc{\pi} t_3} \where {}
\\ \phantom{\, {} = (\{}
   t_1,t_2,t_3 \in \ProcTermr \Land
   t_1 \rbbisim t_1 \altc t_1 \Land t_2 \rbbisim t_2 \altc t_2 \Land
   \encapa(t_1 \altc t_2) \rbbisim \dead\} 
\\ \phantom{\, {} = (}\,
 {} \union
 \{\tup{(\tau \seqc (t_1 \altc t_2) \altc t_1) \paltc{\pi} t_3,
         (t_1 \altc t_2) \paltc{\pi} t_3} \where {}
\\ \phantom{\, {} = (\, {} \union \{}
   t_1,t_2,t_3 \in \ProcTermr \Land
   t_1 \rbbisim t_1 \altc t_1 \Land t_2 \rbbisim t_2 \altc t_2 \Land
   \encapa(t_1 \altc t_2) \rbbisim \dead\}
\\ \phantom{\, {} = (}\,
 {} \union
 \{\tup{\tau \seqc (t_1 \altc t_2) \altc t_1, t_1 \altc t_2} \where {}
\\ \phantom{\, {} = (\, {} \union \{}
   t_1,t_2 \in \ProcTermr \Land
   t_1 \rbbisim t_1 \altc t_1 \Land t_2 \rbbisim t_2 \altc t_2 \Land
   \encapa(t_1 \altc t_2) \rbbisim \dead\}
)\eqvcl\;;
\end{ldispl}%
\item
if $\ax$ is an instance of pBED: similar;
\item
if $\ax$ is an instance 
$\set{X_i = t_i \where i \in I} \Limpl 
 X_j = \rec{X_j}{\set{X_i = t_i \where i \in I}}$ 
($j \in I$) of RSP:
\begin{ldispl}
R_\ax = 
\set{\tup{\theta(X_j),\rec{X_j}{\set{X_i = t_i \where i \in I}}} \where 
\\ \phantom{R_\ax = \{}\, 
     j \in I \Land \theta \in \Theta \Land
     \LAND_{i \in I} \theta(X_i) \rbbisim \theta(t_i)}\eqvcl\;,
\end{ldispl}%
where 
$\Theta$ is the set of all functions from $\cX$ to $\ProcTermr$ and 
$\theta(t)$, where $\theta \in \Theta$ and $t \in \ProcTermr$, stands
for $t$ with, for all $X \in \cX$, all occurrences of $X$ replaced by 
$\theta(X)$.
\end{itemize}
For each equational axiom $\ax$ of \depACPetr, it is easy to check that 
the constructed relation $R_\ax$ is a rooted branching bisimulation 
witnessing, for each closed substitution instance $t = t'$ of $\ax$, 
$t \rbbisim t'$.
For each conditional equational axiom $\ax$ of \depACPetr, it is 
straightforward to check that the constructed relation $R_\ax$ is a 
rooted branching bisimulation witnessing, for each closed substitution 
instance $\set{t_i = t'_i \where i \in I} \Limpl t = t'$ of $\ax$, 
$t \rbbisim t'$ if $t_i \rbbisim t'_i$ for each $i \in I$.
\qed
\end{proof}

The axiom system of \depACPetr\ is incomplete with respect to $\rbbisim$ 
for equations between terms from $\ProcTermr$ and there is no 
straightforward way to rectify this.
The following example shows that the axiom system of \depACPetr\ is 
incomplete.
Let $\sigma \in \EvalMap$ be such that $\sigma(v) = \sigma(w) = 1$.
Then we have that 
\begin{ldispl}
 \eval{\sigma}
  (\ass{v}{v \mul w} \seqc (v = 1 \gc \ass{w}{v}) \parc \ass{v}{v / w}) 
\\ \quad {} \rbbisim
 \eval{\sigma}(\ass{v}{v \mul w} \seqc (v = 1 \gc \ass{w}{v})) \parc 
 \eval{\sigma}(\ass{v}{v / w})\;, 
\end{ldispl}%
but the equation
\begin{ldispl} 
 \eval{\sigma}
  (\ass{v}{v \mul w} \seqc (v = 1 \gc \ass{w}{v}) \parc \ass{v}{v / w}) 
\\ \quad {} =
 \eval{\sigma}(\ass{v}{v \mul w} \seqc (v = 1 \gc \ass{w}{v})) \parc 
 \eval{\sigma}(\ass{v}{v / w}) 
\end{ldispl}%
is not derivable from the axioms of \depACPetr.
This incompleteness cannot be resolved by adding the axiom schema
\begin{ldispl} 
\eval{\sigma}(x \parc y) = \eval{\sigma}(x) \parc \eval{\sigma}(y)
\end{ldispl}%
to the axiom system of \depACPetr\ because the instances of this axiom 
schema are not valid with respect to $\rbbisim$.
The following example shows this.
Let $\sigma \in \EvalMap$ be such that $\sigma(v) = \sigma(w) = 1$.
Then we have that 
\begin{ldispl}
 \eval{\sigma}
  (\ass{v}{v \mul w} \seqc (v = 1 \gc \ass{w}{v}) \parc \ass{v}{v + w}) 
\\ \quad {} \not\rbbisim
 \eval{\sigma}(\ass{v}{v \mul w} \seqc (v = 1 \gc \ass{w}{v})) \parc 
 \eval{\sigma}(\ass{v}{v + w})\;. 
\end{ldispl}%

The following aside is perhaps useful for a better understanding: 
the preceding two examples are related to the notion of interference 
freedom.
Under the given evaluation map $\sigma$, 
$\ass{v}{v \mul w} \seqc (v = 1 \gc \ass{w}{v})$ and $\ass{v}{v / w}$
do not interfere with each other, but 
$\ass{v}{v \mul w} \seqc (v = 1 \gc \ass{w}{v})$ and $\ass{v}{v + w}$ 
do interfere.
This suggest the following characterization of interference freedom.
Let $t,t' \in \ProcTermr$ and $\sigma \in \EvalMap$.
Then $t$ and $t'$ are interference free under $\sigma$ iff
$\eval{\sigma}(t \parc t') \rbbisim
 \eval{\sigma}(t) \parc \eval{\sigma}(t')$.

The axiom system of \depACPetr\ is not even complete in the following 
very limited sense:
\begin{quote}
for all $t, t' \in \ProcTermr$ in which no data parameterized action 
operator, no assignment action operator, and no guarded command operator 
occurs, $t = t'$ is derivable from the axioms of \depACPetr\ if 
$t \rbbisim t'$.
\end{quote}
The origin of this incompleteness is the fact that processes with one or 
more cycles of silent steps, possibly alternated with probabilistic 
choices, can be defined by combining guarded linear recursion and 
abstraction.
Usually, a term denoting such a process is rooted branching bisimulation 
equivalent to a term denoting a process without such cycles, whereas the 
corresponding equation is not derivable from the axioms of \depACPetr.
For example, we have that
\begin{ldispl} 
 \tau \seqc 
 \abstr{\set{b}}(\rec{X}{\set{X = a \paltc{\pi} (a \altc b \seqc X)}})
  \rbbisim
 \tau \seqc a\;,
\end{ldispl}%
but the equation 
\begin{ldispl} 
 \tau \seqc 
 \abstr{\set{b}}(\rec{X}{\set{X = a \paltc{\pi} (a \altc b \seqc X)}}) =
 \tau \seqc a
\end{ldispl}%
is not derivable from the axioms of \depACPetr.
It is likely that this incompleteness can be resolved by adding to the 
axiom system of \depACPetr\ one or more axiom schemas reminiscent to 
the cluster fair abstraction rule added to the axiom system of 
\deACPetr\ in~\cite{Mid21a}.
However, a suitable collection of such axiom schemas have not been found 
yet.

\section{Example: Leader Election}
\label{sect-example}

In this section, \depACPetr\ is used to model the probabilistic leader 
election algorithm for anonymous, bidirectional, asynchronous rings 
given by Bakhshi and others in~\cite{BFPP08a}.
The algorithm is based on a deterministic leader election algorithm for 
onymous, bidirectional, asynchronous rings given by Franklin 
in~\cite{Fra82a}.

First, the mentioned probabilistic version of Franklin's algorithm is 
informally described.

We consider a ring consisting of $n$ processes 
$\PROC{}{0},\ldots,\PROC{}{n-1}$ 
(for some $n \geq 2$) that do not have a unique identity, but know the 
ring size $n$, and $n$ channels between neighboring processes for 
bidirectional message-passing.
The message-passing from one process to a neighboring process is 
asynchronous, and the order of the messages may not be preserved.
A sent message is included in a message container of its destination 
process and will eventually be processed by that process.

Process $\PROC{}{j}$ maintains two (flexible) variables:
\begin{itemize}
\item
$\pid{j} \in \set{1,\ldots,k}$ (for some $k \geq 2$) is the not 
necessarily unique identity of process $\PROC{}{j}$; initially, 
$\pid{j}$ is undefined;
\item
$\bit{j} \in \set{0,1}$ is the number of the current election round of 
process $\PROC{}{j}$ modulo $2$; initially, $\bit{j}$ equals $1$.
\end{itemize}

Each message is a triple $(i,h,b)$ where: 
\begin{itemize}
\item
$i \in \set{1,\ldots,k}$ is the not necessarily unique identity of the 
process from which the message originates;
\item
$h \in \set{1,\ldots,n}$ is a hop counter, which initially has the value 
$1$, and which is increased by one every time the message is forwarded 
by a process;
\item
$b \in \set{0,1}$ is the number of the election round of the process 
from which the message originates modulo $2$.
\end{itemize}

The algorithm proceeds in election rounds, and each process is either 
\emph{active} or \emph{passive}. 
A passive process simply forwards messages after their hop counter has 
been incremented by one. 
Initially, each process is active.

At the start of an election round, each active process $\PROC{}{j}$ 
randomly selects an identity $\pid{j} \in \set{1,\ldots,k}$ and sends 
a message with the selected identity to each of its two neighbours; 
initially, the message is a triple of the form $(\pid{j},1,\bit{j})$. 
Next, $\PROC{}{j}$ receives a message originating from its 
closest active left neighbour and a message originating from its 
closest active right neighbour. 
If the hop counter of both messages is less than the ring size $n$ and 
neither of the messages has a larger identity than its identity, then 
$\PROC{}{j}$ starts a new election round. 
If the hop counter of both messages is less than the ring size $n$ and 
either of those messages has a larger identity than its identity, then 
$\PROC{}{j}$ becomes passive. 
If the hop counter of one of the messages is equal to the ring size $n$, 
$\PROC{}{j}$ becomes the leader and the algorithm terminates.

Thus, during each election round each active process sends two messages 
and receives two messages, whereas each passive process receives and 
forwards two messages.
The repetition of election rounds terminates when an active process 
receives a message from itself.
This is recognized by means of the hop counter of the message, which is 
equal to the ring size $n$ in that case.
Moreover, the fact that the hop counter of the message is equal to the 
ring size $n$ implies that the process is the only active process left 
and that its identity is the largest of the $n$ processes.  

Because the order of the messages may not be preserved, an old message 
may be overtaken by other messages in the ring.
This could in principle lead to a situation in which no single process 
becomes the leader. 
A solution to this problem is to number successive election rounds and 
to supply each message with a round number.
An infinite number of round numbers is not needed because it turns out 
to be sufficient to keep track of round numbers modulo $2$.
The reason for this is that round numbers modulo $2$ suffice to enforce 
message order preservation (see~\cite{BFPP08a}).

In order to use \deACPetr\ for the description of the algorithm, what is 
assumed to be given in Assumptions~\ref{assumption-basic-actions}, 
\ref{assumption-meadow}, \ref{assumption-data}, 
and~\ref{assumption-flex-vars} must actually be given.
It concerns the set $\Act$ of basic actions, the communication function
$\commf$, the signed cancellation meadow $\fM$, the 
signature $\sig_\gD$, the algebra $\gD$, and the sets $\ProgVar_1$, 
\ldots, $\ProgVar_n$ of flexible variables of sort $\Data_1$, \ldots, 
$\Data_n$, respectively.

\begin{sdef}
The set $\Act$ of basic actions is defined as follows:
\begin{ldispl}
\Act =
\Union_{j \in \set{0,\ldots,n-1}} 
\set{\psnd{1}{j},\psnd{2}{j},\chsnd{1}{j},\chsnd{2}{j},
     \chrcv{1}{j},\chrcv{2}{j},\prcv{1}{j},\prcv{2}{j},
     \chcomm{1}{j},\chcomm{2}{j},\pcomm{1}{j},\pcomm{2}{j}} \union 
\set{\leader}
\end{ldispl}%
and the communication function $\commf$ is defined as follows:
\begin{itemize}
\item
for all $j \in \set{0,\ldots,n-1}$:
\begin{ldispl}
\commf(\psnd{1}{j},\chrcv{1}{j}) = \chcomm{1}{j},\;
\commf(\psnd{2}{j},\chrcv{2}{j}) = \chcomm{2}{j},\;
\commf(\chsnd{1}{j},\prcv{1}{j}) = \pcomm{1}{j},\;
\commf(\chsnd{2}{j},\prcv{1}{j}) = \pcomm{2}{j},
\\
\commf(\chrcv{1}{j},\psnd{1}{j}) = \chcomm{1}{j},\;
\commf(\chrcv{2}{j},\psnd{2}{j}) = \chcomm{2}{j},\;
\commf(\prcv{1}{j},\chsnd{1}{j}) = \pcomm{1}{j},\;
\commf(\prcv{1}{j},\chsnd{2}{j}) = \pcomm{2}{j}\;;
\end{ldispl}%
\item
$\commf(a,b) = \dead$ for all other pairs $(a,b) \in \Acttd \x \Acttd$.
\end{itemize}
\end{sdef}

\begin{sdef}
The signed cancellation meadow $\fM$ is the meadow of rational numbers 
(see~\cite{BT07a,BM09g}) expanded with a signum operator.
\end{sdef}

\begin{sdef}
The signature $\sig_\gD$ consists of:
\begin{itemize}
\item
the sorts $\Data_1$, $\Data_2$, and $\Data_3$;
\item
the following constants and operators:
\begin{itemize} 
\item
$\const{0}{\Data_1}$, $\const{1}{\Data_1}$,
$\funct{+}{\Data_1 \x \Data_1}{\Data_1}$, and
$\funct{\bmod}{\Data_1 \x \Data_1}{\Data_1}$;
\item
$\funct{(\ph,\ph,\ph)}{\Data_1 \x \Data_1 \x \Data_1}{\Data_2}$;
\item 
$\const{\mempty}{\Data_3}$,
$\funct{\mset{\ph}}{\Data_2}{\Data_3}$,
$\funct{\munion}{\Data_3 \x \Data_3}{\Data_3}$,
$\funct{\mdiff}{\Data_3 \x \Data_3}{\Data_3}$, and
$\funct{\melem}{\Data_2 \x \Data_3}{\Bool}$
\end{itemize}
\end{itemize}
and the algebra $\gD$ is such that:
\begin{itemize}
\item
the carrier of sort $\Data_1$ in $\gD$ is $\Nat$;
\item 
the carrier of sort $\Data_2$ in $\gD$ is $\Nat \x \Nat \x \Nat$;
\item 
the carrier of sort $\Data_3$ in $\gD$ is the set of all finite 
multisets over $\Nat \x \Nat \x \Nat$;
\item
the interpre\-tation of $0$, $1$, $+$, and $\bmod$ in $\gD$ is as usual;
\item
the interpretation of $(\ph,\ph,\ph)$ in $\gD$ is the operation that 
constructs a triple of natural numbers from three natural numbers;
\item
the interpretation of $\mempty$, $\mset{\ph}$, $\munion$, $\mdiff$, and 
$\melem$ in $\gD$ is a counterpart of the usual set interpretation of 
$\emptyset$, $\set{\ph}$, $\union$, $\diff$, and $\in$, respectively,
that takes into account the multiplicity of the elements of multisets.
The relevant algebraic theory of multisets (also called bags) is given 
in Appendix~\ref{app-multisets}.
\end{itemize}
\end{sdef}

\begin{sdef}
The sets $\ProgVar_1$, $\ProgVar_2$, and $\ProgVar_3$ of flexible 
variables of sort $\Data_1$, $\Data_2$, and $\Data_3$, respectively, are
defined as follows:
\begin{ldispl}
\ProgVar_1 =
\Union_{j \in \set{0,\ldots,n-1}} \set{\pid{j},\bit{j},\lpid{j},\rpid{j}}\;,
\\
\ProgVar_2 = \emptyset\;,
\\
\ProgVar_3 =  
\Union_{j \in \set{0,\ldots,n-1}} \set{\ms{1}{j},\ms{2}{j}}\;.
\end{ldispl}%
\end{sdef}

\begin{snot}
For each $n \in \Nat$, the term $\ul{n}$ of sort $\Data_1$ is defined by 
induction on $n$ as follows: $\ul{0} = 0$ and $\ul{n+1} = \ul{n} + 1$.
\end{snot}

The leader election algorithm for ring size $n$ and process identities 
$\set{1,\ldots,k}$ is described by the following closed \deACPetr\ term:
\begin{ldispl}
\encap{H}(\Parc{j = 1}{n} (\PROC{}{j-1} \parc \CHAN{}{j-1}))\;,
\end{ldispl}%
where, for each $j \in \set{0,\ldots,n-1}$, $\PROC{}{j}$ is defined by 
the recursive specification that consists of the following equations:
\begin{ldispl}
\PROC{}{j} = \ass{\bit{j}}{1} \seqc \PROC{1}{j}\;,
\eqnsep
\PROC{1}{j} = 
\PRC{i = 1}{k}{1/k}
{
\ass{\pid{j}}{\ul{i}} \seqc \psnd{1}{j}((\pid{j},1,\bit{j})) \seqc
\psnd{2}{(j - 1) \bmod n}((\pid{j},1,\bit{j})) \seqc \PROC{2}{j}
}\;,
\eqnsep
\PROC{2}{j} =
\Altc{i = 1}{k} \Altc{h = 1}{n} 
\prcv{1}{(j - 1) \bmod n}((\ul{i},\ul{h},\bit{j})) 
\\ \phantom{\PROC{2}{j} = \Altc{i = 1}{k} \Altc{h = 1}{n}} 
\hsp{-1}
 {} \seqc
(\ul{h} = \ul{n} \gc \leader \altc 
 \Lnot \ul{h} = \ul{n} \gc \ass{\lpid{j}}{\ul{i}} \seqc \PROC{3}{j}) 
\\ \phantom{\PROC{2}{j}} 
 {} \altc
\Altc{i = 1}{k} \Altc{h = 1}{n} 
\prcv{2}{j}((\ul{i},\ul{h},\bit{j})) 
\\ \phantom{\PROC{2}{j} = \Altc{i = 1}{k} \Altc{h = 1}{n}} 
\hsp{-1}
 {} \seqc
(\ul{h} = \ul{n} \gc \leader \altc 
 \Lnot \ul{h} = \ul{n} \gc \ass{\rpid{j}}{\ul{i}} \seqc \PROC{4}{j})\;,
\eqnsep
\PROC{3}{j} =
\Altc{i = 1}{k} \Altc{h = 1}{n} 
\prcv{2}{j}((\ul{i},\ul{h},\bit{j})) 
\\ \phantom{\PROC{3}{j} = \Altc{i = 1}{k} \Altc{h = 1}{n}} 
\hsp{-1}
 {} \seqc
(\ul{h} = \ul{n} \gc \leader \altc 
 \Lnot \ul{h} = \ul{n} \gc \ass{\rpid{j}}{\ul{i}} \seqc \PROC{5}{j})\;,
\eqnsep
\PROC{4}{j} =
\Altc{i = 1}{k} \Altc{h = 1}{n} 
\prcv{1}{(j - 1) \bmod n}((\ul{i},\ul{h},\bit{j})) 
\\ \phantom{\PROC{4}{j} = \Altc{i = 1}{k} \Altc{h = 1}{n}} 
\hsp{-1}
 {} \seqc
(\ul{h} = \ul{n} \gc \leader \altc 
 \Lnot \ul{h} = \ul{n} \gc \ass{\lpid{j}}{\ul{i}} \seqc \PROC{5}{j})\;,
\end{ldispl}
\begin{ldispl}
\PROC{5}{j} = \phantom{\Lnot\!}
(\lpid{j} > \pid{j} \Lor \rpid{j} > \pid{j}) \gc \PROC{6}{j} 
\\ \phantom{\PROC{5}{j}}
 {} \altc
\Lnot (\lpid{j} > \pid{j} \Lor \rpid{j} > \pid{j}) \gc 
\ass{\bit{j}}{(\bit{j} + 1) \bmod \ul{2}} \seqc \PROC{1}{j}\;,
\eqnsep
\PROC{6}{j} = 
\Altc{i = 1}{k} \Altc{h = 1}{n} \Altc{b = 0}{1} 
(
\prcv{1}{(j - 1) \bmod n}((\ul{i},\ul{h},\ul{b})) \seqc
\psnd{1}{j}((\ul{i},\ul{h} + 1,\ul{b})) \seqc \PROC{6}{j} 
\\ 
\phantom{\PROC{6}{j} = \Altc{i = 1}{k} \Altc{h = 1}{n} \Altc{b = 0}{1} (}
\hsp{-1.25}
 {} \altc
\prcv{2}{j}((\ul{i},\ul{h},\ul{b})) \seqc
\psnd{2}{(j - 1) \bmod n}((\ul{i},\ul{h} + 1,\ul{b})) \seqc 
\PROC{6}{j}
)
\;,
\end{ldispl}%
for each $j \in \set{0,\ldots,n-1}$, $\CHAN{}{j}$ is defined by 
the recursive specification that consists of the following equations:
\begin{ldispl}
\CHAN{}{j}  = 
\ass{\ms{1}{j}}{\mempty} \seqc \ass{\ms{2}{j}}{\mempty} \seqc \CHAN{'}{j}
\;,
\eqnsep
\CHAN{'}{j} =
\Altc{i = 1}{k} \Altc{h = 1}{n} \Altc{b = 0}{1} 
\\
\phantom{\CHAN{}{j} = {}} \hsp{1.6}
(
\chrcv{1}{j}(\ul{i},\ul{h},\ul{b}) \seqc 
\ass{\ms{1}{j}}{\ms{1}{j} \munion \mset{(\ul{i},\ul{h},\ul{b})}} \seqc
\CHAN{'}{j} 
\\ 
\phantom{\CHAN{}{j} = {}} \hsp{0.75}
 {} \altc
\chrcv{2}{j}(\ul{i},\ul{h},\ul{b}) \seqc 
\ass{\ms{2}{j}}{\ms{2}{j} \munion \mset{(\ul{i},\ul{h},\ul{b})}} \seqc
\CHAN{'}{j} 
\\ 
\phantom{\CHAN{}{j} = {}} \hsp{0.75}
 {} \altc
(\ul{i},\ul{h},\ul{b}) \in \ms{1}{j} \gc
 \chsnd{1}{j}(\ul{i},\ul{h},\ul{b}) \seqc 
 \ass{\ms{1}{j}}{\ms{1}{j} \mdiff \mset{(\ul{i},\ul{h},\ul{b})}} \seqc
 \CHAN{'}{j} 
\\ 
\phantom{\CHAN{}{j} = {}} \hsp{0.75}
 {} \altc
(\ul{i},\ul{h},\ul{b}) \in \ms{2}{j} \gc
 \chsnd{2}{j}(\ul{i},\ul{h},\ul{b}) \seqc 
 \ass{\ms{2}{j}}{\ms{2}{j} \mdiff \mset{(\ul{i},\ul{h},\ul{b})}} \seqc
 \CHAN{'}{j}
)\;,
\end{ldispl}%
and $H$ is defined as follows:
\begin{ldispl}
H =
\Union_{j \in \set{0,\ldots,n-1}} \Union_{e \in \DataTerm_2}
\\ \qquad \qquad
\set{\psnd{1}{j}(e),\psnd{2}{j}(e),\chsnd{1}{j}(e),\chsnd{2}{j}(e),
     \chrcv{1}{j}(e),\chrcv{2}{j}(e),\prcv{1}{j}(e),\prcv{2}{j}(e)})\;.
\end{ldispl}%

The recursive specifications of $\PROC{}{j}$ and $\CHAN{}{j}$ presented 
above are not guarded linear recursive specifications.
However, they can easily be transformed into guarded linear recursive 
specifications.
The resulting guarded linear recursive specifications of $\PROC{}{j}$ 
and $\CHAN{}{j}$ are given in Appendix~\ref{app-glrs}.
They are less easy to comprehend, but useful when analyzing the 
algorithm based on the axiom system of \deACPetr.

We conclude this section with two claims that will not be proven.

\begin{sclm}
\label{claim-leader-elect-init}
The following equation can be derived from the axioms of \depACPetr\ for
all $\sigma,\sigma' \in \EvalMap$:
\begin{ldispl}
\eval{\sigma}(\encap{H}(\Parc{j = 1}{n}
 (\PROC{}{j-1} \parc \CHAN{}{j-1}))) = 
\eval{\sigma'}(\encap{H}(\Parc{j = 1}{n}
 (\PROC{}{j-1} \parc \CHAN{}{j-1})))\;.
\end{ldispl}%
\end{sclm}
Claim~\ref{claim-leader-elect-init} comes as no surprise if one realizes
that all flexible variables are assigned to before they are used.

\begin{sclm}
\label{claim-leader-elect-correct}
The following holds for all $\sigma \in \EvalMap$:
\begin{ldispl}
\tau \seqc
\abstr{I}(\eval{\sigma}(\encap{H}(\Parc{j = 1}{n} 
 (\PROC{}{j-1} \parc \CHAN{}{j-1})))) \rbbisim 
\tau \seqc \leader\;,
\end{ldispl}%
where $I$ is defined as follows:
\begin{ldispl}
I =
\Union_{j \in \set{0,\ldots,n-1}} 
\\ \phantom{I = {}} \quad \hsp{.25}
(\hsp{.45}
 \Union_{e \in \DataTerm_1}
 \set{\ass{\pid{j}}{e},\ass{\bit{j}}{e},
      \ass{\lpid{j}}{e},\ass{\rpid{j}}{e}}
\\ \phantom{I = {}} \quad
  {} \union
 \Union_{e \in \DataTerm_3}
 \set{\ass{\ms{1}{j}}{e},\ass{\ms{2}{j}}{e}}
\\ \phantom{I = {}} \quad
  {} \union
 \Union_{e \in \DataTerm_2}
 \set{\chcomm{1}{j}(e),\chcomm{2}{j}(e),\pcomm{1}{j}(e),\pcomm{2}{j}(e)}
)\;.
\end{ldispl}%
\end{sclm}
Claim~\ref{claim-leader-elect-correct} states that, when considering all 
assignments and communications that have taken place in the ring during 
the leader election process unobservable, the leader election process, 
after some unobservable activity, always leads to the selection of a 
leader.
This is essentially the same claim as the claim made in~\cite{BFPP08a}.

Because the axiom system of \depACPetr\ is not complete with respect 
to~$\rbbisim$, Claim~\ref{claim-leader-elect-correct} should not be 
confused with the following claim:
\begin{quote}
The following equation can be derived from the axioms of \depACPetr\ for
all  $\sigma \in \EvalMap$:
\begin{ldispl}
\tau \seqc
\abstr{I}(\eval{\sigma}(\encap{H}(\Parc{j = 1}{n} 
 (\PROC{}{j-1} \parc \CHAN{}{j-1})))) = 
\tau \seqc \leader\;.
\end{ldispl}%
\end{quote}

\section{Concluding Remarks}
\label{sect-conclusions}

In this paper, an extension of the imperative process algebra proposed 
in~\cite{Mid21a} with probabilistic choice operators is presented that 
rests on the principle that probabilistic choices are always resolved 
before choices involved in alternative composition and parallel 
composition are resolved.
This extension has been devised, among other things, to be used for 
modeling and analyzing algorithms that are important in the area of 
distributed computing.
Many canonical problems in that area, including the leader election 
problem, call for a probabilistic algorithm.

The suitability of the presented probabilistic imperative process 
algebra for modeling the leader election algorithm for anonymous rings 
given in~\cite{BFPP08a} is demonstrated in this paper.
The suitability of this process algebra for analyzing this leader 
election algorithm is not demonstrated in this paper. 
It is abundantly clear that such analysis is virtually impossible 
without supporting automated tools for this process algebra.

In~\cite{BFPP08a}, it is stated that the probabilistic leader election 
algorithm given in that paper is modeled and analyzed using the $\mu$CRL 
specification language and toolset~\cite{BFGLLP01a}.
Because the $\mu$CRL specification language cannot handle probabilistic 
processes, probabilistic choices had to be replaced by non-deterministic 
choices.
However, it is claimed that it could still be verified with the $\mu$CRL 
toolset that, with probability one, eventually always a leader is 
elected.
I am skeptical of this claim, because it seems to contradict the common 
view that probabilistic choices and non-deterministic choices are 
fundamentally different from each other.
 
In this paper, I build on earlier work on ACP. 
With the exception of axioms A3$'$, A3$''$, CM1E$'$, pBE, and the axioms 
for the probabilistic choice operators, the axioms of \pACPet\ are taken 
from Section 5.3 of~\cite{BW90}.
Axioms A3$'$, A3$''$, CM1E$'$, pBE, and the axioms for the probabilistic 
choice operators are based on work presented in~\cite{Geo11a}.
The axioms for the guarded command operator are basically taken 
from~\cite{BB92c}.
The  evaluation operators are inspired by~\cite{BM05a} and the 
data parameterized action operators are inspired by~\cite{BM09d}.

\appendix

\section{Appendix: An Algebraic Theory of Finite Multisets}
\label{app-multisets}

In this appendix, an algebraic theory of finite multisets over a set is 
presented that is relevant to the example on leader election given in 
Section~\ref{sect-example}.
Multisets are also called bags.

The presented algebraic theory of finite multisets is parameterized by a 
theory of the multiset elements that must satisfy the following 
conditions:
\begin{itemize}
\item
its signature includes a sort $\Bool$ of \emph{booleans} and a sort 
$\Elem$ of \emph{elements};
\item
its signature includes \emph{boolean} constants $\const{\Btrue}{\Bool}$ 
and $\const{\Bfalse}{\Bool}$;
\item
its signature includes a binary \emph{disjunction} operator 
$\funct{\Lor}{\Bool \x \Bool}{\Bool}$;
\item
its signature includes a binary \emph{equality} operator 
$\funct{\Eeq}{\Elem \x \Elem}{\Bool}$;
\item
the following unconditional and conditional equations are derivable from 
its axiom system:
\begin{ldispl}
\begin{eqncol}
\Btrue \Lor \Btrue   = \Btrue\;,                      \\
\Btrue \Lor \Bfalse  = \Btrue\;,                      \\
\Bfalse \Lor \Btrue  = \Btrue\;,                      \\
\Bfalse \Lor \Bfalse = \Bfalse\;,                     
\end{eqncol}
\qquad\quad
\begin{eqncol}
e \Eeq e   = \Btrue\;,                                \\
e \Eeq e'  = e' \Eeq e\;,                             \\
e \Eeq e'  = \Btrue,\, e' \Eeq e'' = \Btrue \Limpl 
e \Eeq e'' = \Btrue\;,
\end{eqncol}
\end{ldispl}%
where 
$e$, $e'$, and $e''$ stand for arbitrary variables of sort $\Elem$.
\end{itemize}

The (parameterized) signature of the parameterized algebraic theory of 
finite multisets consists of the sorts, constants, and operators 
mentioned above and the following sorts, constants, and operators:
\begin{itemize}
\item
a sort $\Mset$ of \emph{multisets};
\item
an \emph{empty multiset} constant $\const{\mempty}{\Mset}$;
\item
a unary \emph{singleton multiset} operator 
$\funct{\mset{\ph}}{\Elem}{\Mset}$;
\item
a binary \emph{multiset union} operator 
$\funct{\munion}{\Mset \x \Mset}{\Mset}$;
\item
a binary \emph{multiset difference} operator 
$\funct{\mdiff}{\Mset \x \Mset}{\Mset}$;
\item
a binary \emph{is-element-of} operator 
$\funct{\melem}{\Elem \x \Mset}{\Bool}$.
\end{itemize}

The (parameterized) axiom system of the parameterized algebraic theory 
of finite multisets consists of the unconditional and conditional 
equations mentioned above and the unconditional and conditional 
equations presented in Table~\ref{eqns-multiset}.
\begin{table}[!t]
\caption{Axioms of finite multisets}
\label{eqns-multiset}
\begin{eqntbl}
\begin{eqncol}
m \munion \mempty = m                                                 \\
m \munion (m' \munion m'') = (m \munion m') \munion m''               \\
m \munion m' = m' \munion m                                      \eqnsep
\mempty \mdiff m = \mempty                                            \\
m \mdiff \mempty = m                                                  \\
m \mdiff (m' \munion m'') = (m \mdiff m') \mdiff m''                  \\
(m \munion \mset{e}) \mdiff \mset{e} = m                              \\
e \Eeq e' = \Bfalse \Limpl 
(m \munion \mset{e}) \mdiff \mset{e'} = 
(m \mdiff \mset{e'}) \munion \mset{e}                            \eqnsep

e \in \mempty = \Bfalse                                               \\
e \in \mset{e} = \Btrue                                               \\
e \Eeq e' = \Bfalse \Limpl 
e \in \mset{e'} = \Bfalse                                             \\
e \in m \munion m' = e \in m \Lor e \in m'                            
\end{eqncol}
\end{eqntbl}
\end{table}
In this table,
$m$, $m'$, and $m''$ stand for arbitrary variables of sort $\Mset$ and
$e$ and $e'$ stand for arbitrary variables of sort $\Elem$.

Here, we are only interested in the models of an instantiation of this 
theory of multisets that satisfy the following conditions:
\begin{itemize}
\item
for each sort from the signature of the instantiation, each member of 
the carrier of that sort can be represented by a closed term of that 
sort;
\item
for each sort from the signature of the instantiation, each equation 
between closed terms of that sort holds only if it is derivable from the 
axiom system of the instantiation.
\end{itemize}
In other words, we are only interested in initial models.

It is easy to develop a theory of triples of natural numbers that 
satisfies the conditions imposed above on a theory of multiset elements.
Now take the actualization of the signature $\sig_\gD$ and the algebra 
$\gD$ described in Section~\ref{sect-example}.
Then the carrier of the sorts from that actualization of $\sig_\gD$ and 
the interpretation of the constants and operators from that 
actualization of $\sig_\gD$ intended in Section~\ref{sect-example} are 
provided by the initial model of the instantiation of the presented 
theory of finite multisets with that theory of triples of natural 
numbers.

In publications on multisets, the multiset union operator $\munion$ is 
also called additive union operator, addition operator, and sum 
operator.
This is usually done in order not to confuse it with another counterpart 
of the set union operator $\union$ 
(see e.g.~\cite{GM96b,JGT01a,Ade10a}).
That other operator is then called maximal union operator or union 
operator.
Using as usual the symbol $\union$ for that operator, it can be defined
in terms of $\munion$ and $\mdiff$ as follows:
$m \union m' = (m \mdiff m') \munion m'$ (see~\cite{Ade10a}).
On the other hand, the operator $\munion$ cannot be defined in terms of 
$\union$ and $\mdiff$ (see also~\cite{Ade10a}).

In~\cite{Bli89a}, a general first-order theory of (finite and infinite) 
multisets is developed that extends Zermelo-Fraenkel set theory (with 
the axiom of choice) to multisets.
The algebraic theory of finite multisets presented in this appendix is 
in accordance with that theory.

\section{Guarded Linear Recursive Specifications of $\PROC{}{j}$ and
         $\CHAN{}{j}$}
\label{app-glrs}

In this appendix, the linearized version of  the recursive 
specifications of $\PROC{}{j}$ and $\CHAN{}{j}$ presented in 
Section~\ref{sect-example} are given.

Let $E$ be a recursive specification over \depACPet, and
let $E'$ be a guarded linear recursive specification over \depACPet\ 
such that $\vars(E) \subseteq \vars(E')$.
Then, under rooted branching bisimulation equivalence, $E$ has a unique
solution and, for all $X \in \vars(E)$, the $X$-component of the unique 
solution of $E$ is the same as $X$-component of the unique solution of 
$E'$ if the following conditions are satisfied for all $X \in \vars(E)$:
\begin{itemize}
\item
 $\rec{X}{E'} = \rec{t}{E'}$ is derivable from the axioms of \depACPetr;
\item
$E \Limpl X = \rec{X}{E'}$ is derivable from the axioms of \depACPetr.
\end{itemize}
This follows immediately from the soundness of the axiom system of
\depACPetr\ with respect to rooted branching bisimulation equivalence.

It is easy to see that the above conditions are satisfies in the case 
of the recursive specification of $\PROC{}{j}$ presented in 
Section~\ref{sect-example} and its linearized version presented below
and in the case of the recursive specification of $\CHAN{}{j}$ presented 
in Section~\ref{sect-example} and its linearized version presented 
below.

For each $j \in \set{0,\ldots,n-1}$, $\PROC{}{j}$ is defined by the 
guarded linear recursive specification that consists of the following 
equations:
\begin{ldispl}
\begin{aeqns}
\PROC{}{j} & = & \True \gc \ass{\bit{j}}{1} \seqc \PROC{1}{j}\;,
\eqnsep
\PROC{1}{j} & = & 
\PRC{i = 1}{k}{1/k}
{\,\True \gc \ass{\pid{j}}{\ul{i}} \seqc \PROC{1.1}{j}}\;,
\eqnsep
\PROC{1.1}{j} & = &
\True \gc \psnd{1}{j}((\pid{j},1,\bit{j})) \seqc \PROC{1.2}{j}\;,
\eqnsep
\PROC{1.2}{j} & = & 
\True \gc \psnd{2}{(j - 1) \bmod n}((\pid{j},1,\bit{j})) \seqc 
\PROC{2}{j}\;,
\eqnsep
\PROC{2}{j} & = & 
\Altc{i = 1}{k} \Altc{h = 1}{n} 
(\hsp{.25}
 \True \gc \prcv{1}{(j-1) \bmod n}((\ul{i},\ul{h},\bit{j})) \seqc
 \PROC{2.1}{j,i,h} 
\\ & & \phantom{\Altc{i = 1}{k} \Altc{h = 1}{n}} \hsp{-.5}
  {} \altc
 \True \gc \prcv{2}{j}((\ul{i},\ul{h},\bit{j})) \seqc
 \PROC{2.2}{j,i,h})\;,
\eqnsep
\PROC{2.1}{j,i,h} & = & 
\ul{h} = \ul{n} \gc \leader \seqc \PROC{7}{j} \altc 
\Lnot \ul{h} = \ul{n} \gc \ass{\lpid{j}}{\ul{i}} \seqc \PROC{3}{j}
\\ & & \quad
\text{for}\; i \in \set{1,\ldots,k}\;
\text{and}\; h \in \set{1,\ldots,n}\;, 
\eqnsep
\PROC{2.2}{j,i,h} & = & 
\ul{h} = \ul{n} \gc \leader \seqc \PROC{7}{j} \altc 
\Lnot \ul{h} = \ul{n} \gc \ass{\rpid{j}}{\ul{i}} \seqc \PROC{4}{j}
\\ & & \quad
\text{for}\; i \in \set{1,\ldots,k}\;
\text{and}\; h \in \set{1,\ldots,n}\;, 
\eqnsep
\PROC{3}{j} & = &
\Altc{i = 1}{k} \Altc{h = 1}{n} 
\True \gc \prcv{2}{j}((\ul{i},\ul{h},\bit{j})) \seqc \PROC{3.1}{j,i,h}\;,
\eqnsep
\PROC{3.1}{j,i,h} & = &
\ul{h} = \ul{n} \gc \leader \seqc \PROC{7}{j} \altc 
\Lnot \ul{h} = \ul{n} \gc \ass{\rpid{j}}{\ul{i}} \seqc \PROC{5}{j})
\\ & & \quad
\text{for}\; i \in \set{1,\ldots,k}\;
\text{and}\; h \in \set{1,\ldots,n}\;, 
\eqnsep
\PROC{4}{j} & = &
\Altc{i = 1}{k} \Altc{h = 1}{n} 
\True \gc \prcv{1}{(j-1) \bmod n}((\ul{i},\ul{h},\bit{j})) \seqc 
\PROC{4.1}{j,i,h}\;,
\eqnsep
\PROC{4.1}{j,i,h} & = &
\ul{h} = \ul{n} \gc \leader \seqc \PROC{7}{j} \altc 
\Lnot \ul{h} = \ul{n} \gc \ass{\lpid{j}}{\ul{i}} \seqc \PROC{5}{j})
\\ & & \quad
\text{for}\; i \in \set{1,\ldots,k}\;
\text{and}\; h \in \set{1,\ldots,n}\;, 
\eqnsep
\PROC{5}{j} & = & \phantom{\Lnot}
(\lpid{j} > \pid{j} \Lor \rpid{j} > \pid{j}) \gc \tau \seqc \PROC{6}{j} 
\\ & \altc &
\Lnot (\lpid{j} > \pid{j} \Lor \rpid{j} > \pid{j}) \gc 
\ass{\bit{j}}{(\bit{j} + 1) \bmod \ul{2}} \seqc \PROC{1}{j}\;,
\eqnsep
\PROC{6}{j} & = & 
\Altc{i = 1}{k} \Altc{h = 1}{n} \Altc{b = 0}{1} 
(\hsp{.25}
 \True \gc \prcv{1}{(j-1) \bmod n}((\ul{i},\ul{h},\ul{b})) \seqc 
 \PROC{6.1}{j,i,h,b}
\\ & & \phantom{\Altc{i=1}{k} \Altc{h=1}{n} \Altc{b=0}{1}} \hsp{-.5}
  {} \altc 
 \True \gc \prcv{2}{j}((\ul{i},\ul{h},\ul{b})) \seqc
 \PROC{6.2}{j,i,h,b})\;,
\end{aeqns}
\end{ldispl}%
\begin{ldispl}
\begin{aeqns}
\PROC{6.1}{j,i,h,b} & = & 
\True \gc \psnd{1}{j}((\ul{i},\ul{h} + 1,\ul{b})) \seqc \PROC{6}{j} 
\\ & & \quad
\text{for}\; i \in \set{1,\ldots,k}, h \in \set{1,\ldots,n},\;
\text{and}\; b \in \set{0,1}\;, 
\eqnsep
\PROC{6.2}{j,i,h,b} & = & 
\True \gc \psnd{2}{(j-1) \bmod n}((\ul{i},\ul{h} + 1,\ul{b})) \seqc 
\PROC{6}{j}
\\ & & \quad
\text{for}\; i \in \set{1,\ldots,k}, h \in \set{1,\ldots,n},\;
\text{and}\; b \in \set{0,1}\;, 
\eqnsep
\PROC{7}{j} & = & \True \gc \ep\;.
\end{aeqns}
\end{ldispl}%

For each $j \in \set{0,\ldots,n-1}$, $\CHAN{}{j}$ is defined by the 
guarded linear recursive specification that consists of the following 
equations:
\pagebreak[2]
\begin{ldispl}
\begin{aeqns}
\CHAN{}{j} & = & 
\True \gc \ass{\ms{1}{j}}{\mempty} \seqc \CHAN{1}{j}\;,
\eqnsep
\CHAN{1}{j} & = &
\True \gc \ass{\ms{2}{j}}{\mempty} \seqc \CHAN{'}{j}\;,
\eqnsep
\CHAN{'}{j} & = &
\Altc{i = 1}{k} \Altc{h = 1}{n} \Altc{b = 0}{1} 
(\hsp{.25}
 \True \gc \chrcv{1}{j}(\ul{i},\ul{h},\ul{b}) \seqc \CHAN{'\!1}{j} 
\\ & & \phantom{\Altc{i=1}{k} \Altc{h=1}{n} \Altc{b=0}{1}} \hsp{-.5} 
 {} \altc
 \True \gc \chrcv{2}{j}(\ul{i},\ul{h},\ul{b}) \seqc \CHAN{'\!2}{j} 
\\ & & \phantom{\Altc{i=1}{k} \Altc{h=1}{n} \Altc{b=0}{1}} \hsp{-.5}
 {} \altc
(\ul{i},\ul{h},\ul{b}) \in \ms{1}{j} \gc
 \chsnd{1}{j}(\ul{i},\ul{h},\ul{b}) \seqc \CHAN{'\!3}{j} 
\\ & & \phantom{\Altc{i=1}{k} \Altc{h=1}{n} \Altc{b=0}{1}} \hsp{-.5} 
 {} \altc
(\ul{i},\ul{h},\ul{b}) \in \ms{2}{j} \gc
 \chsnd{2}{j}(\ul{i},\ul{h},\ul{b}) \seqc \CHAN{'\!4}{j}
)\;,
\eqnsep
\CHAN{'\!1}{j} & = &
\True \gc 
\ass{\ms{1}{j}}{\ms{1}{j} \munion \mset{(\ul{i},\ul{h},\ul{b})}} \seqc
\CHAN{'}{j}\;, 
\eqnsep
\CHAN{'\!2}{j} & = &
\True \gc 
\ass{\ms{2}{j}}{\ms{2}{j} \munion \mset{(\ul{i},\ul{h},\ul{b})}} \seqc
\CHAN{'}{j}\;,
\eqnsep
\CHAN{'\!3}{j} & = &
\True \gc 
\ass{\ms{1}{j}}{\ms{1}{j} \mdiff \mset{(\ul{i},\ul{h},\ul{b})}} \seqc
\CHAN{'}{j}\;, 
\eqnsep 
\CHAN{'\!4}{j} & = &
\True \gc 
\ass{\ms{2}{j}}{\ms{2}{j} \mdiff \mset{(\ul{i},\ul{h},\ul{b})}} \seqc
\CHAN{'}{j}\;.
\end{aeqns}
\end{ldispl}%

\bibliographystyle{splncs04}
\bibliography{PA}

\end{document}